%% file: groupoids_and_fields_INDAM_2024.tex
\documentclass[graybox]{svmult}

\usepackage{type1cm}        % activate if the above 3 fonts are
\usepackage{makeidx}         % allows index generation
\usepackage{graphicx}        % standard LaTeX graphics tool
\usepackage{multicol}        % used for the two-column index
\usepackage[bottom]{footmisc}% places footnotes at page bottom

\usepackage{newtxtext}       % 
\usepackage[varvw]{newtxmath}       % selects Times Roman as basic font

\makeindex             % used for the subject index
\usepackage[english]{babel}
\usepackage{tikz}
\usepackage{tikz-cd}

\usepackage{xcolor}
\usepackage[verbose]{hyperref}
\hypersetup{colorlinks=false,allbordercolors=blue,pdfborderstyle={/S/U/W 1}}

\usepackage{amsmath}

\usepackage{amssymb}
\usepackage{mathrsfs}
\usepackage{empheq}

\usepackage{doi}
\renewcommand{\doi}[1]{}

\begin{document}

\title*{``Fields'' in classical and quantum field theories}
\author{Alberto Ibort\orcidID{0000-0002-0580-5858}, Arnau Mas\orcidID{0000-0003-0532-0938} and\\ Luca Schiavone\orcidID{0000-0002-1817-5752}}
\institute{Alberto Ibort, Arnau Mas \at Department of Mathematics, Univ. Carlos III de Madrid, Avda. de la Universidad 30, 28911, Legan'es, Madrid, Spain, \email{albertoi@math.uc3m.es, arnau.mas@icmat.es} and ICMAT, Calle Nicolas Cabrera, 13-15. Campus de Cantoblanco, 28049 Madrid, Spain,
\and Luca Schiavone \at Dipartimento di Matematica e Applicazioni, Universit\`a Federico II di Napoli, Via Cintia, Monte S. Angelo I-80126, Napoli, Italy \email{luca.schiavone@unina.it}}
\maketitle
\begin{abstract}  
\newline The challenges posed by the development of field theories, both classical and quantum, force us to question their most basic and foundational ideas like the role and origin of space-time, the meaning of physical states, etc.  Among them the notion of ``field'' itself is notoriously difficult to address.   These notes aim to analyze such notion from the perspective offered by the groupoid description of quantum mechanics inspired by Schwinger's picture of quantum mechanics.  Then, a natural interpretation of the notion of physical fields as functors among appropriate groupoids will emerge. The domain of a field in this new picture is a groupoid that describes ``test particles'', and its codomain is a groupoid that describes the intrinsic nature of the system being probed.    Such a space of functors carries some natural structures, which are best described in a categorical language.   Some illustrative examples will be presented that could help clarify the various abstract notions discussed in the text. 
 \end{abstract}
 
%\tableofcontents


%%%%%%%%%%%
%%%%%%%%%%%

\section{Introduction}\label{sec:introduction}

In these notes a major conceptual revision of the notion of physical fields is proposed.   This new interpretation is inspired on the groupoid picture of quantum mechanics that has been recently proposed (see, for instance, \cite{Ci19a,Ci19b,Ci19c,Ci19d,Ci24} and references therein).   The fundamental idea supporting this new formulation is that the description of a physical system is often done by probing it using simpler systems, called ``test particles'' in the jargon of field theories.   This idea is analised in the conceptual framework provided by the aforementioned groupoid picture of quantum mechanics and from there the main notions involved in the theory of fields will be derived, among them the notion of field itself that must be modified accordingly, becoming functors, or groupoid homomorphisms, among groupoids rather than functions among sets.

Many physical theories use the notion of ``field'' as a main tool to take into account the observations that take place in a background spacetime.     In simple classical situations such fields are often defined in terms of more fundamental physical quantities that are assumed to be well determined.   For instance, the field $\mathbf{v} = \mathbf{v}(x,y,z,t) \in \mathbb{R}^3$ of velocities used to describe a fluid is well determined from the physical characteristics of the fluid itself.    More challenging is to ascertain the nature of classical fields associated to fundamental interactions like the electric or magnetic fields, $\mathbf{E} = \mathbf{E}(x,y,z,t)$, $\mathbf{B} = \mathbf{B}(x,y,z,t)$, created by a distribution of charges or, even harder, the electromagnetic potential $A = (A_0, \mathbf{A})$ which is used as the fundamental physical entity in quantum electrodynamics.  The same can be said of Newtonian's gravitational potential field $\phi = \phi (x,y,z,t)$, or Einstein's metric field $g = g_{\mu\nu} (x) dx^\mu dx^\nu$.    In the later cases, the physical and mathematical characteristics of the given field are tested or probed by using ``test particles'', that is, idealized systems that interact with the field without perturbing it.    Thus we describe physical fields in an indirect way by probing them with auxiliary systems whose characteristics are used to construct a mathematical picture of the theory.

\textit{Classical fields and the mathematical description of classical fields.}  In most cases, the experimental devices required to observe the test particles used in the description of the field itself, require a spacetime framework, that is, the events describing the behaviour of the test particles take place in a given spacetime domain (typically provided by the given laboratory framework).   Thus, the mathematical structure that emerges in the classical treatment of field theories is that physical fields are modelled as functions $\phi$ defined on a given spacetime $\mathscr{M}$ with values in a set carrying additional structures that captures the observations performed on the given test particles.  Such structures are often linear structures, or Lie algebras of Lie groups, other manifolds, etc.    More generally, physical fields are modelled mathematically as cross sections $\phi$ of a given bundle $\pi \colon \Omega \to \mathscr{M}$:
\begin{equation}\label{eq:field_sections}
\phi \colon \mathscr{M} \to \Omega \, , \qquad \pi \circ \phi = \mathrm{id\,}_{\mathscr{M}} \, ,
\end{equation}
thus, locally, fields look just as functions on open subsets $U \subset \mathscr{M}$, with values in a given set $V$ (the standard fibre of the bundle).    In addition to all this physical fields are subjected to transform properly with respect to the symmetries of the theory:  Kinematical symmetries like Poincar\'e symmetry in relativistic field theories or gauge transformations when describing fundamental interactions.     This mathematical framework to describe theories of fields is widely acknowledged today, however the exact physical nature of the bundle $\pi \colon \Omega \to \mathscr{M}$ is not completely clear in spite of its usefulness. 

The geometrical structures which are present in the spaces of fields used to describe classical theories, for instance in the space of jets of the bundle $\pi \colon \Omega \to \mathscr{M}$, or in the  bundle describing the covariant Hamiltonian picture of the theory, become of paramount importance and a lot of efforts have been poured in trying to ascertain their main conceptual ingredients.   It is impossible in this brief introduction to give due credit to the many ideas and historical developments that have led to the well established geometrical pictures of field theories used today.  We will just point out that the Lagrangian description of field theories has reached a mature state whose geometry is neatly encoded in the so called variational bicomplex (see, for instance, \cite{De99} and references therein).   More recently a categorical description of Lagrangian field theories has been gaining traction (see, for instance \cite{De99,Bl24}).  A comprehensive treatment of Lagrangian field theories using the categorical notion of the topos of smooth sets is presented in \cite{Gi23}.  Alternatively, a covariant Hamiltonian picture has been firmly established under the name of the multisymplectic picture of field theories (see, for instance \cite{Ca91,Go04,Ib17} and references therein) and its implications still being developed.    

\textit{Quantum fields.} Very early after the dawn of quantum mechanics the problem of describing relativistic quantum systems brought the introduction of ``quantum'' versions of the dynamical equations of classical field theories.   The first attempts gave rise to relativistic wave equations, like Dirac's equation, and the fields of the theory were  again sections of bundles over spacetime, in the case of Dirac's equation, sections of a spinor bundle over Minkowski spacetime.    Quite soon it was realised that a consistent description of the quantum structure of the observed systems implied a quantization of the given ``classical fields'', a process that was known as ``second quantization''.  Classical fields were turned to field operators: $\phi \mapsto \Phi$, where $\Phi$ now represents a distribution on Minkowski space with values on bounded operators on a given Hilbert space  $\mathcal{H}$ (supporting a representation of Poincar\'e's group).  The formalization of all these ideas led to the axiomatic description of quantum field theories provided by the Streater-Wightman axioms \cite{St00} and, after that other axiomatic formalizations like Osterwalder-Schrader axiomatic Euclidean field theories (see, for instance, \cite{Gl87}).

\textit{Feynman's picture:  Lagrangian relativistic field theories.}  In spite of the formal developments around the axiomatic mathematical foundations of quantum field theories discussed in the previous paragraph, R. Feynman's path integral formulation of Lagrangian field theories opened a complete new way of thinking about quantum field theories.   In Feynman's picture fields are still classical fields but Feynman's dynamical principle determines the physical amplitudes of the theory as the path integral on the space of classical fields of the exponential of the classical action of the theory.  In spite of its intrinsic technical difficulties, the success of the perturbative treatment of electrodynamics justified such approach and it has become the main tool in the study and analysis of fundamental interactions ever since (see, for instance, \cite{Ra89}).

\textit{Schwinger's picture of quantum mechanics and field theories.}  In an extraordianry tour de force J. Schwinger established a Lagrangian picture of quantum field theories whose fields were operators instead of classical functions and the dynamical principle of the theory was given in terms of a quantum formulation of Hamilton's variational principle that, in a classical setting, is known as the Schwinger-Weiss variational principle \cite{We36,Sc91}.   In short, it establishes that the variation of the physical amplitudes are given by the expected value of the variation of an operator action functional and depends solely on the action of operators at boundary states.   In the particular instance of quantum mechanics, if $\langle b, t_1 \mid a, t_0\rangle$ denotes the amplitude of reaching the quantum state determined by $b$ at time $t_1$ when the system was at the state determined by $a$ at time $t_0$, then:
$$
\delta \langle b, t_1 \mid a, t_0\rangle = i \langle b, t_1 \mid  \delta \int_{t_0}^{t_1} \mathbf{L}(s) ds \mid a, t_0\rangle = \langle b, t_1 \mid G_2 - G_1\mid a, t_0\rangle \, ,
$$
where $\mathbf{L} (s)$ is an operator Lagrangian functional, and $G_1,G_2$ are Hermitean operators acting on the given Hilbert space of the theory.   A clever use of the previous ideas led to the well-known Dyson-Schwinger equations that can be used to compute the desired physical amplitudes.    We would like to stress here that both approaches, that of Feynman and of Schwinger, are widely different, the first and most important difference is the completely different use of the notion of fields, classical functions in the case of Feynman, operators in the case of Schwinger.    Later on, Schwinger tried to provide a new foundation to the theory of quantum fields upon a more solid basis introducing the primary notion of \textit{sources} that, combined with the notion of particles, allow to derive the notion of fields, not as a primary concept but as a ``book-keeping'' device to account for the propagation of the effects created by sources \cite{Sc98}.

\textit{Other approaches to a theory of quantum fields.}   Feynman and Schwinger's pictures are not by any means the only pictures that have been proposed to describe a relativistic theory of interactions.    We will just mention here the program started by R. Haag \cite{Ha12} aiming to provide a natural foundation for such theories using $C^*$-algebras of operators instead of fields and known today as Algebraic quantum field theories (AQFT) (see also \cite{Fr06} where the theory is formalized as a functor between categories, and references therein) and A. Connes attempts to provide a consistent description of the Standard Model by using non-commutative geometry \cite{Co19}.    It is important to realize that in these two approaches the notion of field itself is increasingly blurred and becomes subsidiary.   

Our point of view is that in spite of the remarks in the last paragraph, the notion of field is of fundamental relevance to understand the inner structure of a physical theory of interactions and that many difficulties that arise when using them, that is, when we have to combine the notion of classical field (like in Feynman's picture) with a standard quantum mechanical interpretation based on linear operators in Hilbert spaces, come from a poor understanding of the auxiliary structures used in the construction of the theory and, hitherto in the a wrong understanding of the role and nature of fields.     In this paper we propose a new understanding of the notion of physical fields which is intrinsically ``quantum'', that is, the notion of field is derived from a quantum description of the elements used in the construction of the theory, namely, the notion of ``test particles'', that will be called ``probing systems'', and how do they interact with the system under study.  In this sense fields are nothing but the histories of such interactions as described by probing systems.  The mathematical formalization of such notion of histories become functors (or homomorphisms) among groupoids, the domain of the functor the probing system, and codomain the groupoid describing the intrinsic structure of the system we are studying.   In the particular case that the probing system is ``classical'', that is, its corresponding algebra is Abelian, then, we will recover the standard (classical) definition of fields as functions defined on some auxiliary space.     In general, if the groupoid describing the probing system is non-classical, i.e., its algebra is non-commutative, then a field defined on such system, that is, a functor with domain such groupoid, describes a genuine quantum field.   Examples of various situations illustrating all these possibilities will be discussed along the text.

The paper will be organized as follows.   Sect. \ref{sec:functors} will be devoted to succinctly discuss the basic ideas on the groupoid picture of quantum mechanics that will be needed in what follows, and the notion of probing systems as the groupoidal formulation of the notion of ``test particles''.     These ideas will immediately give rise to the notion of ``quantum histories'' as functors among groupoids, hence the dictum that ``fields are functors'' among groupoids.  Sect. \ref{sec:functors} will end up discussing some simple examples that illustrate the scope and richness of the previous ideas.  In Sect. \ref{sec:category} the natural properties of groupoids describing probing systems will be exploited to argue that local fields must be identified with sheaves on a natural category associated to the groupoid describing the given probing system.   Then, we get immediately that the category of sheaves on probing groupoids is a topos and the fundamental gauge symmetry of physical theories is identified with natural transformations among functors describing physical fields.   Finally, after discussing more relevant examples that include gauge theories, Sect. \ref{sec:higher} will be devoted to introduce a new composition of local fields, called horizontal composition, and how this operation is promoted to a 2-category structure on the topos of local fields.   

%%%%%%%%%%%
%%%%%%%%%%%

\section{Fields are functors}\label{sec:functors}
The natural question that emerges from the discussion in the Introduction is: Is there a natural interpretation of physical fields in a pure quantum mechanical environment?  In other words, is there a way of identifying some or all of the many notions of fields discussed in the introduction as emerging from purely quantum mechanical notions?  In this section we will try to give a positive answer to this question by reflecting on the groupoid descripton of the simplest example of physical field provided by Feynman's path integral interpretation of quantum mechanics and the role played by clocks and test particles.

\subsection{Groupoids and fields: probing systems}

We will begin by succinctly reviewing the basic ideas on the groupoids picture of quantum mechanics.

\subsubsection{Groupoids and quantum systems}\label{sec:quantum}

Inspired by the early ideas of W. Heisenberg \cite{He25} and later on, J. Schwinger \cite{Sc59,Sc91}, it was proposed that the abstract description of the experimental setting describing a quantum system is provided by a groupoid \cite{Ci19a}.    Thus if the outputs of observations are given by elements $a,b,\dots$ of a set $\Omega$, and the transitions that the system can experiment are denoted as $\alpha \colon a \to b$, $\beta \colon c \to d$, etc.\footnote{In Schwinger's formalism, transitions $\alpha\colon a \to b$ correspond to selective measurements $M(b,a)$ \cite{Sc91}, that is, given a physical system a selective measurement is a process such that immediately before we will always get the value $a$ it the observable $A$ would be measured, and that immediately afterwards, we will always get the value $b$ would the observable $B$ be measured.}, they will be identified with the morphisms of a groupoid $\Gamma$ whose partial composition law will correspond to the concatenation of transitions.    The source and the target maps $s,t \colon \Gamma \to \Omega$, will be given by $s (\alpha) = a$, $t (\alpha) = b$, for $\alpha \colon a \to b$. Two transitions $\alpha$, $\beta$ that can be composed, that is, such that $t(\alpha) = s(\beta)$, will be said to be composable, and the set of all such pairs will be denoted by $\Gamma^{(2)} \subset \Gamma \times \Gamma$.   The composition of the transition $\alpha \colon a \to b$ and the transition $\beta \colon b \to c$, will be denoted as $\beta \circ \alpha \colon a \to c$, and it satisfies the associativity property:  $(\gamma \circ \beta) \circ \alpha = \gamma \circ (\beta \circ \alpha)$ provided that $s(\gamma) = t(\beta)$, and $s(\beta) = t (\alpha)$.  Moreover it will be assumed that for any outcome $a \in \Omega$, that is for any object of the groupoid $\Gamma$, there is a transition $1_a \colon a \to a$, which is neutral, that is $\alpha\circ 1_a = \alpha$, and $1_b \circ \alpha = \alpha$.      The existence of a transition $\alpha^{-1} \colon b \to a$ such that $\alpha^{-1} \circ \alpha = 1_a$, and $\alpha \circ \alpha^{-1} = 1_b$ reflects the \textit{principle of microscopic reversibility} as stated by R. Feynman \cite[p. 4]{Fe05}\footnote{Which is arguably the stronger of the assumptions we are making.}.    Thus, under the previous assumptions, the family of transitions $\alpha\colon a \to b$ of a given quantum system form an algebraic groupoid, denoted in what follows $\Gamma \rightrightarrows \Omega$ (emphasizing the source and target maps on the space of outcomes).  Many examples describing groupoids relevant to the description of simple quantum systems are provided in the literature (see, for instance, \cite{Ci19a,Ci19e,Ci19f,Ci24}).      

Given a groupoid $\Gamma \rightrightarrows \Omega$, the set of transitions $\alpha \colon a \to a$, for a given $a \in \Omega$, form a group, called the isotropy group of $\Gamma$ at $a$, and denoted by $\Gamma (a)$.  The collection of all isotropy groups 
\begin{equation}\label{eq:fundamental}
\Gamma_0 = \bigsqcup_{a \in \Omega}\Gamma(a) \, ,
\end{equation}
 defines a normal subgroupoid of $\Gamma$, called the fundamental subgroupoid (see, for instance, \cite{Ib19} for basic notions on groupoids).   The set of all transitions with source $a \in \Omega$, will be denoted $\Gamma_a$, and the transitions with target $b\in \Omega$ will be denoted $\Gamma^b$.  Thus, $\Gamma (a) = \Gamma_a \cap \Gamma^a$.    Given $a \in \Omega$, the orbit $\Gamma a \subset \Omega$, consists of all elements $b\in \Omega$ such that there is $\alpha \colon a \to b$.   The space of orbits of the groupoid $\Gamma$ will be denoted by $\Omega/\Gamma$.   If the groupoid has a single orbit, that is, for any $a,b\in \Omega$, there is a transition $\alpha \colon a \to b$, we say that the groupoid is transitive or connected.

Often we have to consider additional mathematical structures in our description of quantum systems, for instance, we may consider that the groupoid $\Gamma$ is a measure groupoid (see, for instance, \cite{Ci24}).   In such case the groupoid $\Gamma$ carries a class $[\nu]$ of measures consistent with the algebraic properties of the groupoid and it has associated a canonical von Neumann algebra $\nu (\Gamma)$ supported on the Hilbert space $L^2(\Gamma, \nu)$.    Such algebra is identified in a natural way with the ``algebra of observables'' of the theory and relates the groupoidal picture with the standard $C^*$-algebraic description of quantum mechanical systems.      As it will be discussed further in this paper, the description of various physical systems of interest involve the introduction of additional mathematical structures, most important smooth and topological structures.     A natural way to do that is to consider Lie, topological or  diffeological groupids (or even better, smooth sets groupoids).    Lie groupoids will appear in our presentation of kinematical groupoids (see Sect. \ref{sec:kinematical}) and later on, when stepping further in the categorical background for field theories, we will introduce other specific examples of topological groupoids.

\begin{remark}\label{rem:groupoid_class}    As indicated before, when we say that a groupoid $\Gamma \rightrightarrows \Omega$ describes a quantum system we are not being specific about the class of additional structures it carries, i.e., we are not being specific whether the groupoid $\Gamma$ is a topological groupoid, a Lie groupoid or it just carries a suitable Borel structure.    The different choices will depend on the specific properties of the system under study.  Thus, in some instances the groupoid $\Gamma$ will just be a finite or discrete countable groupoid, or in other cases it could be a Lie groupoid.   Later on (see Sect. \ref{sec:yang_mills}) we will see that a class of diffeological groupoids play a relevant role in the description of Yang-Mills theories.   Thus, in what follows,  we will not specify the class of the groupoid we are working with, and we will make this explicit in dealing with specific situations.
\end{remark}

\subsubsection{Clocks and test particles: probing systems}\label{sec:probing}

The key idea discussed in this section is that of a ``probing system'', our formal expression of the notion of ``test particle''.  In physical terms, a test particle is a system that interacts with the field we are studying and that is negligible in the sense that its presence does not affect the properties of the system being observed.   The notion of a ``test particle'' is thus an idealization of a class of experimental procedures that, more formally, can be understood as the introduction of an auxiliary system that will form part of the experimental setting, i.e., of the overall groupoid describing the full system as explained below, but that nevertheless it will not disturb the original system we were interested in.  Let us formalize these notions to get a suitable definition that could be stated in neat mathematical terms.  

The auxiliary system must be such that it provides accurate outcomes describing relevant properties of the system under study\footnote{For instance, the orientation of ferrite sards dropped on a magnetic field or the position and momentum of a test charged particle in an electric field.} and that, in spite of this, it will not affect the original system itself described by a groupoid $\Gamma \rightrightarrows \Omega$.   The first property will be taken into account by saying that if $\mathscr{P} \rightrightarrows \Sigma$ is the groupoid describing our ``probing system'', then there must exist a groupoid homomorphism $D \colon \widetilde{\Gamma} \to \mathscr{P}$, where $\widetilde{\Gamma}$ denotes the groupoid describing the total system under study (including the auxiliary system itself), such that for every possible transition $\alpha \colon a \to b$, the outcomes $x = D(a)$, $y = D(b) \in \Sigma$ are accurately recorded by the device $\mathscr{P}$\footnote{Think, for instance, of a test particle moving under the influence of a electromagnetic field, and the outcomes $x$ are the recordings of the position of the particle at a given time $t$.}.   The groupoid homomorphism property captures the consistency of the readings of the probing system with the composition of transitions of the total system, i.e., 
\begin{equation}\label{eq:detection}
D(\alpha \circ \beta) = D(\alpha) \circ D(\beta) \, , \qquad D(1_a) = 1_x \, .
\end{equation}     

The independence property, i.e., that the system $\mathscr{P}$ does not affect the original system $\Gamma$, is partially captured by observing that the probing system $\mathscr{P}$ is neccesarily interacting with the system we want to study, i.e., it itself is part of the total system $\widetilde{\Gamma}$, that will include both $\Gamma$ and $\mathscr{P}$ as subsystems, and thus, its own transitions $\sigma \colon x \to y \in \mathscr{P}$,  can be considered to be transitions of the total system $\widetilde{\Gamma}$.   We may conclude that there is an embedding of the groupoid $\mathscr{P}$ as a subgroupoid of $\widetilde{\Gamma}$, denoted as $W \colon \mathscr{P} \to \widetilde{\Gamma}$, that identifies each transition $\sigma \colon x \to y$ of the probing system $\mathscr{P}$, with a transition $W(\sigma) \colon W(x) \to W(y)$.    The two homomorphisms of groupoids $D$ and $W$ must be consistent in the sense  that $D \circ W = \mathrm{id}_\mathscr{P}$, or, in other words, the readings provided by the probing system of the transitions of the probing system as transitions in the total system must  be the same as the readings of the probing system as an standalone system.     

In the same way we can argue that the original system $\Gamma$ is a subsystem of the total system $\widetilde{\Gamma}$, that is, it defines a subgroupoid of $\widetilde{\Gamma}$ and the two subgroupoids $\Gamma$ and $W(\mathscr{P})$ must generate the full groupoid and have minimal intersection, that is 
\begin{equation}\label{eq:minimal}
\Gamma \cap W(\mathscr{P}) = 1_\Sigma \, ,
\end{equation}
where $1_\Sigma$ denotes the trivial groupoid defined by the set $\Sigma$ itself, that is, the groupoid whose objects are the elements of the set $\Sigma$ and whose only transitions are the units $1_x \colon x \to x$ (such groupoid will also be called the unit groupoid defined by the set $\Sigma$).  Equation (\ref{eq:minimal}) also implements the aforementioned requirement that ``the outcomes of the probing system must describe relevant properties of the system'' $\Gamma$.    The previous discussion can be summarized in the diagram below, Fig. (\ref{fig:short}).  We have now all the ingredients to provide a formal definition of a probing system.

\begin{figure}\begin{center}
\begin{tikzpicture} 
\draw [->] (0.4,0) -- (1.8,0);
\draw [->] (2.2,0) -- (3.8,0);
\draw [->] (4.2,0) -- (5.8,0);
\draw [->] (6.2,0) -- (7.8,0);
\draw [->] (2.2,-1.5) -- (3.8,-1.5);
\draw [->] (4.2,-1.5) -- (5.8,-1.5);
\draw [->] (1.9,-0.2) -- (1.9,-1.2);
\draw [->] (3.9,-0.2) -- (3.9,-1.2);
\draw [->] (5.9,-0.2) -- (5.9,-1.2);
\draw [->] (2.1,-0.2) -- (2.1,-1.2);
\draw [->] (4.1,-0.2) -- (4.1,-1.2);
\draw [->] (6.1,-0.2) -- (6.1,-1.2); 
\node  at (0,0) {$1_\Omega$};
\node  at (2,0) {$\Gamma$};
\node(tilGam) at (4,0) {$\widetilde{\Gamma}$};
\node(P)  at (6,0) {$\mathscr{P}$};
\node  at (8,0) {$1_*$};
\node  at (2,-1.5) {$\Omega$};
\node(Omega2)  at (4,-1.5) {$\Omega$};
\node(Sigma)  at (6,-1.5) {$\Sigma$};
\node  at (3,0.2) {$i$};
\node  at (5,0.2) {$D$};
\node  at (5,0.8) {$W$};
\node  at (5,-1.3) {$\pi$};
\node  at (5,-2.3) {$\phi$};
\draw[->, dashed] (Sigma) to [bend left=45](Omega2);
\draw[->, dashed] (P)  to [bend right=45] (tilGam);
\end{tikzpicture}
\caption{Short exact sequence describing a probing system as a subsystem of the total system}\label{fig:short}
\end{center}
\end{figure}
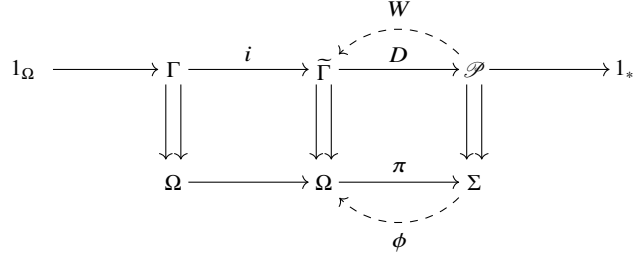

\begin{definition}\label{def:probing}
A probing system for the quantum system described by the groupoid $\Gamma \rightrightarrows \Omega$ is a quantum system described by the groupoid $\mathscr{P} \rightrightarrows \Sigma$ such that there exists a total groupoid $\widetilde{\Gamma} \to \Omega$ containing $\Gamma$ as a subgroupoid, a functor $D \colon \widetilde{\Gamma} \to \mathscr{P}$, called the ``detection'' functor, which is surjective, and a family of injective functors $W \colon \mathscr{P} \to \widetilde{\Gamma}$ which are right inverses of $D$, i.e., $D \circ W = \mathrm{id}_\mathscr{P}$, that describe the different ways that $\mathscr{P}$ can be identified with a subgroupoid of $\widetilde{\Gamma}$, that is, as a subsystem of $\widetilde{\Gamma}$, and such that the short sequence of functors in the upper row of diagram (\ref{fig:short}) is exact, that is, $\Gamma$ is a normal subgroupoid of $\widetilde{\Gamma}$ and $\widetilde{\Gamma}/\Gamma \cong \mathscr{P}$.
\end{definition}

A natural example of the previous ideas is provided by the notion of ``clock''.    A clock is just an auxiliary system used as a reference system to tag the evolution of the system under observation.    We assume that the outputs of such auxiliary system are just real numbers $t$ that will be identified with the physical time as measured by the reference system.   We may think of these outputs as numbers in the real line $\mathbb{R}$ or, alternatively, as a discrete set, for instance the integers $\mathbb{Z}$, or  as a periodic quantitity, like elements of $\mathbb{Z}_n$.  Thus the abstract description of a clock will be given by a groupoid $\mathscr{C} \rightrightarrows \mathbb{R}$.       Note that we are not concerned about the specific dynamics of the clock.    All that we care about is the accuracy of its readings, i.e., the real numbers $t_0,t_1$ associated to any observed transition $\alpha \colon a \to b$ of the total system $\Gamma$, $t_0 = D(a)$, $t_1 = D(b)$.     We can assume that in any experimental setting there is a clock (a master clock) that has been set up and synchronised with all the devices used in the observations and experiments that are carried on.      The groupoid homomorphism $W \colon \mathscr{C} \to \widetilde{\Gamma}$ defining $\mathscr{C}$ as a subgropoid of the total system is determined by the process of synchronisation that has been carried on among the various apparatuses involved in the experimental setting.     Arguably the simplest model for a clock $\mathscr{C}$ in the groupoid picture will be provided by the groupoid of pairs $\mathscr{C} = \mathbb{R} \times \mathbb{R} \rightrightarrows \mathbb{R}$, with composition law $(s,t) \circ (t, r) = (s,r)$.  

According to Def. \ref{def:probing} the probing system $\mathscr{C}$ determines the short exact sequence of groupoids (\ref{fig:short}):
\begin{equation}\label{eq:sequence}
1_\Omega \to \ker D = \Gamma \to \widetilde{\Gamma} \stackrel{D}{\to} \mathscr{C} \to 1_* \, ,
\end{equation}
where $1_\Omega$ and $1_*$ denote the canonical groupoid over the set $\Omega$ and the point $*$ respectively.    The cross section homomorphism $W \colon \mathscr{C} \to \tilde{\Gamma}$ splits the sequence  (\ref{eq:sequence}).    The existence of such splitting doesn't imply the triviality of the sequence itself, i.e., the groupoid $\widetilde{\Gamma}$ is not necessarily the product groupoid $\widetilde{\Gamma} = \Gamma \times \mathscr{C}$ (see \cite{Ma05,Ib23} for more details on the theory of extensions of groupoids).   In fact, the detection homomorphism $D$ defines a projection $\pi \colon \Omega \to \Sigma$, and the homomorphism  $W$ defines a cross section $\phi$ of the projection  map $\pi$ (see Fig. \ref{fig:short}).    Thus, our total system $\widetilde{\Gamma} \rightrightarrows \widetilde{\Omega}$ is not just a trivial product groupoid, in which case $\widetilde{\Omega} = \Omega \times \Sigma$, with $\Omega$ the space of objects of the groupoid $\Gamma$.   Instead, as indicated before and shown in diagram \ref{fig:short}, the total groupoid $\widetilde{\Gamma}$ has space of objects $\Omega$ and has the form $\Gamma \times_\Omega \pi^*\mathscr{P}$, where $\pi^* \mathscr{P}$ denotes the pull-back along $\pi \colon \Omega \to \Sigma$ of the groupoid $\mathscr{P} \rightrightarrows \Sigma$, that is 
$$
\pi^*\mathscr{P} = \{(b, \sigma, a) \in \Omega\times \mathscr{P} \mid \sigma \colon x \to y, \pi (a) = x, \pi (b) = y \} \, ,
$$ 
and $H \times_\Omega K$ denotes the fibrered product of the groupoids $H, K \rightrightarrows \Omega$, that is:   
$$
H\times_\Omega K = \{ (\alpha, \sigma) \in H \times K \mid s(\alpha) = s(\sigma), t(\alpha) = t(\sigma) \} \,.
$$

Note that $\Gamma \times_\Omega \pi^* \mathscr{P}$  can be identified with pairs of transitions $(\alpha, \sigma) \in \Gamma \times \pi^*\mathscr{P}$ such that the projections of the outcomes of $\alpha$ are the same as the outcomes of $\sigma$.    Under such circumstances we will say that the total system groupoid $\widetilde{\Gamma}$ is an almost product of the groupoids $\Gamma$ and $\mathscr{P}$ along the map $\pi \colon \Omega \to \Sigma$.

Now, a cross section homomorphism $W \colon \mathscr{P} \to \widetilde{\Gamma} = \Gamma \times_\Omega \pi^* \mathscr{P}$ can be identified with a homomorphism of groupoids $W_\Gamma \colon \mathscr{P} \to \Gamma$ by means of the formula:
$$
W_\Gamma (\sigma) = \mathbf{pr}_1 (W(\sigma)) \, ,
$$
with $\mathbf{pr}_1 \colon \Gamma \times_\Omega \pi^*\mathscr{P} \to \Gamma$, the canonical projection in the first factor.  Equivalently,
 $W(\sigma) = (W_\Gamma(\sigma), (\phi (y) ,\sigma, \phi (x)))$, $x = s(\sigma)$, $y = t(\sigma)$.    In what follows we will just refer to the groupoid homomorphism $W_\Gamma \colon K \to \Gamma$, instead of the cross section $W$ and we will use the same notation for both.     Finally, we observe that the detection map $D \colon \widetilde{\Gamma} \to \mathscr{P}$ is given by $D (\alpha, \sigma) = \sigma$, if $\widetilde{\Gamma}$ is the almost product of $\Gamma$ and $\mathscr{P}$, and that $\Gamma = \ker D$.
 
 %%%%%%%%%%%%
 %%%%%%%%%%%%
 
 \subsection{Feynman's quantum mechanics: paths, clocks and histories}\label{sec:feynman}

We can illustrate the previous ideas in the context of Feynman's quantum mechanics.   Feynman's path integral description of quantum mechanics can be understood as a field theory in 1D where the auxiliary spacetime of the theory is one-dimensional and can be identified with a clock.    The standard fields of the theory are paths, that is maps $\gamma \colon [t_0,t_1] \subset \mathbb{R} \to Q$, where $Q$ is the space of classical configurations of the system under scrutiny.    Such paths can be understood as (local\footnote{Any local path can be trivially extended to a global section.}) sections of the trivial bundle $\pi \colon E = Q \times \mathbb{R} \to \mathbb{R}$.   However, in our view, to be consistent with the groupoidal description of quantum mechanics, the ``history'' of the system is determined by a clock, that in its simplest conception can be described by the groupoid $\mathscr{C}$ of pairs $(t,s) \in \mathbb{R}\times \mathbb{R}$, also denoted $\mathscr{C} := P(\mathbb{R}) = \mathbb{R} \times \mathbb{R} \rightrightarrows \mathbb{R}$.    Then, a history would be a functor $W \colon P(\mathbb{R}) \to \Gamma$.   The inner groupoid structure of the system can be chosen to be again the groupoid of pairs, $\Gamma = P(Q) := Q \times Q\rightrightarrows Q$.  The interpretation of such groupoid is that the transitions of the system $\alpha \colon q_0 \to q_1$, are just pairs $(q_1,q_0)$, where $q_0$ will be the configuration of the system before the transition and $q_1$ the configuration at the end of the transition.   Then the functor $W$ assigns a configuration $q(t) := W(t)$ to each $t$, and a transition $W(\alpha) \colon q(t) \to q(s)$ for $\alpha \colon t \to s$.   Hence there is a natural one-to-one correspondence between functors $W \colon \mathscr{C} \to P(Q)$ and paths $\gamma \colon \mathbb{R} \to Q$.    Moreover, any functor $W \colon \mathscr{C} \to P(Q)$ can be restricted to any subgroupoid $\mathscr{C}' \subset P(\mathbb{R})$, $W_{\mathscr{C}'} \colon \mathscr{C}' \to P(Q)$, then we may consider the collection of all histories defined on finite intervals $[t_0,t_1]$ as restrictions of histories defined on the total real line.   We will denote by $P[t_0,t_1]$ the subgroupoid of pairs $[t_0,t_1]\times [t_0,t_1]$.

Note that histories can be composed in a natural way, that is, if $W_1 \colon P[t_0,t_1] \to P(Q)$,  $W_2 \colon P[t_1,t_2] \to P(Q)$ are two histories, then we can define the history 
$$
W_2 \circ W_1 \colon P[t_2,t_0] \to P(Q) \, ,
$$ 
which is the compositon of them that results from concatenation of the corresponding paths, that is: 
\begin{equation}\label{eq:composition_feynman}
W_2 \circ W_1 (t,s) = \left\{ \begin{array}{ll} W_1(t,s) \, , & \mathrm{if\, }  t_0 \leq t,s \leq t_1 \, , \\ W_2(t,t_1) \circ W_1(t_1,s) \, , & \mathrm{if\,} t_0 \leq s \leq t_1 \leq t \leq t_2 \\ W_2(t,s) \, , & \mathrm{if\, }  t_1 \leq t,s \leq t_2
\end{array} \right. \, .
\end{equation}
    Thus the space of histories becomes a category with this composition law.  Later on (see Sect. \ref{sec:higher}) we will see how these notions extend to a much more general realm (see also \cite{Ci24} for a detailed analysis of Feynman's path integral in the groupoid picture of quantum mechanics).

%%%%%%%%%%%%

\subsection{Kinematical groupoids}\label{sec:kinematical}

The probing system $\mathscr{P}\rightrightarrows \Sigma$ can have additional kinematical properties, that is, in addition to be a ``clock'', it could provide readings of spatial coordinates in our laboratory.  Thus, the space of outcomes $\Sigma$ of the system could be identified with a subset of spacetime.    Even more, we can think that the role of the probing system includes the determination of the spacetime structure we are using to support the description of the system under consideration.   In fact, when we impose consistency conditions for the synchronisation of various probing systems around our laboratory, we introduce a kinematical structure on the space $\Sigma$ that leads to Einstein's kinematical model of spacetime as a Lorentzian manifold, e.g. Minkowski space in the case of Einsteinian relativistic systems\footnote{A detailed discussion of these notions in the groupoidal setting will be presented elsewhere.}.    

Thus we can provide a first definition of a \textit{relativistic groupoid} as a probing system whose space of objects is a suitable relativistic spacetime.      In such case we can, as in the case of a simple clock discussed above, consider that our relativistic groupoid is just the groupoid of pairs $\mathscr{P} = \Sigma \times \Sigma \rightrightarrows \Sigma$ with the standard composition law.   However the kinematical structure of the space of objects (outcomes) carries the geometrical structure provided by the principle of relativity we are using, i.e., in the case of Minkowski space $\Sigma = \mathbb{M}^4$, the constancy of the velocity of light in all inertial frames, that translates in the invariance of the theory under the action of Poincar\'es group $\mathcal{P}$ on $\mathbb{M}^4$.    In such case the appropriate kinematical groupoid would be the action groupoid $\mathcal{P} \times \mathbb{M}^4 \rightrightarrows \mathbb{R}^4$, whose source and target maps are given by $s((\Lambda, a), x) = x$, $t((\Lambda, a), x) = \Lambda x + a$, where $(\Lambda ,a) \in \mathcal{P}$, $\Lambda$ denotes a Lorentz transformation and $a\in \mathbb{R}^4$ a translation, and $x \in \mathbb{M}^4$ is a point in Minkowski space.   The composition in the action groupoid is given by the standard composition law $((\Lambda',a'),x') \circ ((\Lambda,a),x) = ((\Lambda'\Lambda,\Lambda' a + a'), x)$, with $x' = \Lambda x + a$.

Of course, such a definition of kinematical groupoid can be extended trivially to the family of spacetimes defined by the family of kinematical relativistic groups classified by Bacry and Levy-Leblond \cite{Ba68}, however it would be more appealing, and at the end much more interesting, to take the alternative road provided by the natural identification of the kinematical groupoid over Minkowski space $\mathscr{P}_{\mathrm{Minkowski}} = \mathcal{P} \times \mathbb{M}^4$ with the Riemann groupoid\footnote{The so called ``Riemann groupoid'' in these notes, is sometimes refered as the ``groupoid of frames'' \cite{Ma05} in the literature.} associated to the spacetime $\mathbb{M}^4$, that is:
\begin{equation}\label{eq:riemann}
\mathcal{R}(\mathbb{M}^4) = \{ (y,\Lambda_{yx}, x) \mid x,y \in \mathbb{M}^4, \Lambda_{yx} \colon T_x\mathbb{M}^4 \to T_y \mathbb{M}^4, \Lambda_{yx}^* {\eta_0}(y) = \eta_0(x) \} \, ,
\end{equation}
where in the previous definition $\eta_0$ denotes the canonical Lorenztian metric on $\mathbb{M}^4$, that is $\eta _0 = - (dt)^2 + (dx)^2 + (dy)^2 + (dz)^2$.  Even more, we would like to encode in the kinematical groupoid describing our spacetime the ``symmetries'' of the spacetime $\Sigma$, but such symmetries cannot be described in general by an action groupoid constructed from an isometry group (as in Poincar\'e's case) because such isometry group will be trivial generically, but in spite of this, the ``symmetries'' of the system are encoded in a groupoid, the Riemann groupoid of the given spacetime (see, for instance, \cite{We96,Ib19,Ib25b}).
Now, this will allow us to define the kinematic groupoid associated with any spacetime $(\mathscr{M}, \eta)$ as the Riemann groupoid $\mathcal{R}(\mathscr{M}) \rightrightarrows \mathscr{M}$ defined by a formula similar to (\ref{eq:riemann}).

Thus in what follows we will consider probing systems $\mathscr{P}\rightrightarrows \mathscr{M}$ that provide accurate readings of spacetime events $x\in\mathscr{M}$.   Such systems can be described by simple kinematical groupoids like the ones consider above or more elaborate constructions as it will be discussed in the coming sections.    In any case, the possible embeddings of our probing system in the system we are studying will be described by groupoid homomorphisms (or functors) $W \colon \mathscr{P} \to \Gamma$ where the groupoid $\Gamma \rightrightarrows \Omega$ will describe the ``inner structure'' of the quantum system in contraposition with the spacetime structure description provided by the probing system $\mathscr{P}\rightrightarrows \mathscr{M}$.  Note the the assumption of the existence of probing systems implies the existence of a bundle $\pi \colon \Omega \to \mathscr{M}$ as an \textit{a priori} notion of our theory, justifying in this way the use of bundles to model fields basic in any theory of fields as discussed in the Introduction. 

In the current setting the functors $W \colon \mathscr{P} \to \Gamma$ will be called generalised histories (or just ``histories'') for short, and they encapsulate both the notion of a classical field as a section $\phi$ of the projection bundle map $\pi$, that is, $\phi (x) := W(x)$, $x \in \mathscr{M}$, Eq. (\ref{eq:field_sections}), and the notion of relativistic transformation (or symmetry), that is, given a morphism $\sigma \colon x \to y \in \mathscr{P}$, then $W(\sigma) \colon \phi (x) \to \phi (y)$.   In other words, the field $\phi$ associated to the history $W$ transforms covariantly with respect to the ``relativistic symmetries'' of the given spacetime.    If we use a notation closer to the standard notation of sets and maps to indicate the morphisms on the groupoids, that is we denote as $y = \sigma (x)$\footnote{We aware that here $\sigma$ is not a map acting on the manifold $\mathscr{M}$ by isometries, it is just a morphism in the kinematical gropoid $\mathscr{P}$.} the morphism $\sigma \colon x \to y$, and $\phi (y ) = W(\sigma) (\phi (x))$, then the functor/groupoid homomorphism property reads:
\begin{equation}\label{eq:covariant}
\phi (\sigma (x) ) = W(\sigma)(\phi (x)) \, .
\end{equation}

%%%%%%%%%%%%

\subsection{Functor Fields: ``classical'' vs. ``quantum fields''}\label{sec:functor_fields}

We can summarize the discussion in the previous sections by establishing the following definition that identifies the notion of physical fields and quantum histories.    Because the ideas and arguments that have led us to this notion are based upon the groupoidal description of quantum mechanics, the ultimate interpretation of the notion of fields introduced here is quantum.  In this sense we  might call them ``quantum fields''.  However, we must point out here that no ``quantization'' scheme is involved in their definition and, as it will be shown immediately afterwards, the standard ``classical'' fields are particular instances of them.   

In order to avoid unnecessary terminology conflicts we will chose to call the new notion of fields ``functor fields'', even though they are intrinsically quantum objects as they are functors among quantum systems.

\begin{definition}\label{def:quantum field}
A functor field (or just a field) is a functor $W$ from the category determined by a probing groupoid $\mathscr{P}\rightrightarrows \mathscr{M}$ into a groupoid $\Gamma \rightrightarrows \Omega$.   The groupoid $\Gamma$ will be called the internal quantum groupoid of the theory.
\end{definition}

As hinted in the Introduction, the distinction between ``classical'' and ``quantum'' fields can be understood through the nature of the probing groupoid $\mathscr{P}$.

\bigskip

\textit{``Classical Fields'':}  If the probing groupoid $\mathscr{P} \rightrightarrows \mathscr{M}$ is ``classical," meaning that it effectively reduces to the space of objects $\mathscr{M}$\footnote{Even, as it was mentioned before, a classical system is defined by a groupoid whose algebra is Abelian, which implies that the groupoid is essentially to a groupoid of the form $\underline{\mathscr{M}}\times A$ with $A$ an Abelian group \cite{Ci24}.}, called before the unit groupoid over $\mathscr{M}$ and denoted by $1_\mathscr{M} = \underline{\mathscr{M}}$, Sect. \ref{sec:probing}, then a functor $W \colon \mathscr{P} \to \Gamma$ becomes a map from the spacetime $\mathscr{M}$ to the objects $\Omega$ of $\Gamma$. If $\Gamma = Q \times Q \rightrightarrows Q$ as in the Feynman example, $W$ just picks up a configuration $q(x)$ for each $x \in \mathscr{M}$. More generally, if $\pi \colon \Omega \to \mathscr{M}$ is the projection associated with $\Gamma$, a functor $W$ from the unit groupoid $\underline{\mathscr{M}}$ to $\Gamma$ requires $W(x) = 1_{W(x)}$ for each $x \in \mathscr{M}$, essentially selecting an object $W(x) \in \Omega$ over each $x$. We recover the notion of a classical field as a section $\phi: \mathscr{M} \to \Omega$ where $\phi(x) = W(x)$, Eq. (\ref{eq:field_sections}), (see also Fig. \ref{fig:short}).

Thus a simple classical scalar field $\phi \colon \mathscr{M} \to \mathbb{R}$ fits Definition \ref{def:quantum field}, with the spacetime $\mathscr{M}$ being identified with the unit groupoid $\underline{\mathscr{M}} \rightrightarrows \mathscr{M}$
 (which represents a purely classical probing system that only identifies points in spacetime without inherent transitions between them); the internal quantum groupoid will be the unit groupoid $\Gamma = \underline{\mathbb{R}} \rightrightarrows \mathbb{R}$, whose objects are real numbers and whose only morphisms are identities $1_r \colon r \to r$.   A functor $W \colon \mathscr{P} \to \Gamma$ must map objects $x \in \mathscr{M}$ to objects $W(x) \in \mathbb{R}$, and identity morphisms $1_x$ in $\mathscr{P}$ to identity morphisms $W(1_x) = 1_{W(x)}$ in $\Gamma$. The map $x \mapsto W(x)$ is precisely a function $\phi \colon \mathscr{M} \to \mathbb{R}$. The functor condition $W(1_x) = 1_{W(x)}$ is automatically satisfied for any such function.

Thus, as an extreme situation, we can consider again the set $\Omega$ as a groupoid $\Gamma = \underline{\Omega}$.  Note that in this case both the source and target maps $s,t$ coincide and they can be identified with the identity map.   Such groupoid is totally disconnected and, in physical terms, there are no transitions (or ``quantum jumps'') among its elements.  Such groupoid portraits a classical system defined on the set $\Omega$ and, provided an ancillary Borel structure and a reference state, i.e., a probability measure on it is chosen, its algebra of observables is the Abelian von Neumann algebra $\nu (\Omega) = L^\infty(\Omega)$.    It can be shown that classical systems will always be defined in these terms (see \cite{Ci24,Ci20}). 

Summarizing, if our probing system is classical in the previous sense, it will be described by a set-groupoid $\mathscr{M}$, that we may identify with a (local) spacetime structure determined by our experimental setting.  Then a functor $W \colon \underline{\mathscr{M}} \to \underline{\Omega}$, will be just a map $\phi \colon \mathscr{M} \to \Omega$.  On the other hand, the fact that $\underline{\mathscr{M}}$ is an actual probing system for our theory, implies that there is a detection functor $D \colon \underline{\Omega} \to \underline{\mathscr{M}}$, where we are now assuming that $\Omega$ captures all the aspects involved in the classical description, including the spacetime itself.  Thus, the detection map $D$ becomes just a projection map $\pi \colon \Omega \to \mathcal{M}$, and the functor $W$ becomes a section of such map or, in other words, a  standard classical field.  

\bigskip

%%%%%%%%%%%%%
%%%%%%%%%%%%%

\textit{``Quantum Fields'':}  On the other hand if the probing groupoid $\mathscr{P}$ is ``non-classical'', for instance, if it involves non-commutative structures (like groupoids arising from non-commutative algebras or quantum groups) or encodes non-trivial quantum transitions, the resulting functor field $W \colon \mathscr{P} \to \Gamma$ becomes inherently quantum. The functoriality condition $W(\sigma \circ \rho) = W(\sigma) \circ W(\rho)$ imposes non-trivial constraints reflecting the quantum nature of the probing system's composition law. The field $W$ is no longer just a function on spacetime but maps the transitions/symmetries of the probing system $\mathscr{P}$ to transitions/symmetries in the internal groupoid $\Gamma$.  Non-trivial examples of functor fields defined on non-trivial probing groupoids will be considered later on (see Sect. \ref{sec:top_gauge}, \ref{sec:yang_mills}).

Thus, functor fields defined on probing systems which are not classical, have an intrinsic quantum flavour, however, this non-commutative structure could also be tied to symmetries of the system, like in the discussion of kinematical groupoids, Sect. \ref{sec:kinematical}.  Moreover, functor fields do not require any \textit{ad hoc} quantization method.  They are naturally defined on the given quantum systems, that is, groupoids, used in the theory.   However we can recover the Streater-Wightman axiomatics by using a representation of the theory.  Such matters will be discussed at length in further works.

% (Further examples could follow, e.g., gauge fields, simple QM systems)

%%%%%%%%%%
\subsection{Functor fields and representations of groupoids}\label{sec:representations}

The concept of functor is very rich and can be looked up in different ways.   In what follows we are going to exploit some of these possibilities to understand better some of the implications of the previous definition.   We will begin by focusing first in functors as representations and, in the coming section, we will look at fields as sheaves on categories.

A functor $F \colon \mathsf{C} \to \mathsf{D}$ defined on the category $\mathsf{C}$ and taking values in the category $\mathsf{D}$ can be thought as a representation of $\mathsf{C}$ on $\mathsf{D}$.   For instance, if the category $\mathsf{C}$ is a group $G$, that is, we think of a group $G$ as a category $\mathsf{G}$ with just one object $\star$ and such that all its morphisms $g \colon \star \to \star$ are invertible, then a functor $R \colon \mathsf{G} \to \mathsf{Vect}$, where $\mathsf{Vect}$ denotes the category of complex linear spaces, is in fact a linear representation of $G$.   The functor $R$ assigns a linear space $V$ to the object $\star$, and a linear map $R(g) \colon V \to V$, to any $g \in G$, in such a way that 
\begin{equation}\label{eq:representation_group}
R(gg') = R(g)R(g') \, , \quad  \forall g,g' \in G \, , \quad R(e) = \mathrm{id}_V \, .
\end{equation}
In the same way if $\mathsf{Hilb}$ denotes the category of Hilbert spaces, that is the category whose objects are complex separable Hilbert spaces and whose morphisms are bounded linear operators, a unitary representation $U$ of $G$ will be a functor $U \colon \mathsf{G} \to \mathsf{Hilb}$, whose codomain is the subgroupoid of unitary operators among Hilbert spaces.    Additional topological properties could be asked to the representations under consideration depending on the topological properties of the group $G$, for instance we can demand that $R$ is strongly continuous, that will be specified in particular instances.

Thus, we may think of a functor $W \colon \mathscr{P} \to \Gamma$, as a representation of the groupoid $\mathscr{P}$ in the category defined by the groupoid $\Gamma$ itself\footnote{Recall that a groupoid is a (small) category all whose morphisms are invertible.}.   In general a functor $R \colon~\mathscr{P} \to \mathsf{Set}$ will be a representation of the groupoid $\mathscr{P}\rightrightarrows M$ in the category of sets.  Such functor can be thought as a representation of $\mathscr{P}$ on the groupoid associated to a surjective map $\pi \colon \Omega \to M$, where we associate to each  $x\in M$ the set $R(x) = \pi^{-1}(x) = \Omega_x$, the fibre of $\pi$ at $x$; and to each morphism $\alpha \colon x \to y\in \mathscr{P}$ an invertible map $R(\alpha) \colon R(x) = \Omega_x \to R(y) = \Omega_y$, such that 
\begin{equation}\label{eq:representation_groupoid}
R(\alpha \circ \beta) = R(\alpha) R(\beta) \, , \quad  \forall (\alpha, \beta) \in \Gamma^{(2)} \, , \quad R(1_x) = \mathrm{id}_{\Omega_x} \, ,
\end{equation}
in agreement with (\ref{eq:representation_group}).   The disjoint union of all sets $R(x) = \Omega_x$ will be:
$$
\Omega = \bigsqcup_{x \in M} \Omega_x \, ,
$$ 
and the projection $\pi \colon \Omega \to M$, will be the canonical map $\pi (\xi) = x$, for any $\xi \in \Omega_x$.   Note that associated to any surjection $\pi \colon \Omega \to M$, there is a groupoid, denoted by $\mathrm{Aut \, }(\pi)$ in what follows, whose morphisms are bijections $\varphi_{yx} \colon \Omega_x \to \Omega_y$, and the composition is the standard composition of maps.  In this sense a representation $R$ of the groupoid $\mathscr{P}$ in the category of sets is just a functor $R \colon \mathscr{P} \to \mathrm{Aut\,}(\pi)$.   

Similarly to the case of groups, a linear representation of the groupoid $\mathscr{P}$ will be a functor $R \colon \mathscr{P} \to \mathsf{Vect}$ or, a functor $R \colon \mathscr{P} \to \mathrm{Aut\,}(\pi)$, where $\pi \colon \Omega \to M$ will be a projection whose fibres are linear spaces.   If the groupoid $\mathscr{P}$ carries additional topological structures, for instance is a locally compact topological groupoid, we may ask the representation, and in consequence the map $\pi \colon \Omega \to M$ to carry an additional topological structure that will make the functor $R$ continuous.    A natural request  is that $\pi \colon \Omega \to M$ defines a topological vector bundle and the map 
$$
R \colon \mathscr{P} * \Omega \to \Omega \, ,  \quad R(\alpha, \xi) := R(\alpha)\xi \, , \quad (\alpha, \xi) \in \mathscr{P}* \Omega \, ,
$$ 
is continuous, where $\mathscr{P} * \Omega \subset \mathscr{P} \times \Omega$, denotes all pairs $(\alpha, \xi)$ which are compatible, i.e., such that $\xi \in \Omega_{s(\alpha)}$.

A family of functors fields, relevant as physical fields, would be functors $W \colon \mathscr{P} \to \mathsf{Set}_G$, i.e., functors defined on a given groupoid $\mathscr{P}$ with values in the category of $G$-sets, that is, the category whose objects are sets $A$ carrying an action $(g,a) \mapsto ga$, $g \in G$, $a\in A$, of the group $G$, and whose morphisms are $G$-equivariant maps $\varphi \colon  A \to B$, $\varphi (g a) = g \varphi (a)$, for all $g\in G$, $a \in A$.    As indicated before any such functor can be thought of as a functor $W \colon \mathscr{P} \to \mathrm{Aut\,}_G(\pi)$, where $\mathrm{Aut\,}_G(\pi)$ denotes the groupoid of automorphisms of the projection map $\pi \colon \Omega \to M$, but now each map $\varphi_{yx} \colon \Omega_x \to \Omega_y$ is equivariant with respect the action of $G$ on each fibre $\Omega_x$ of $\pi$.    Thus, a functor field $W \colon \mathscr{P} \to~\mathsf{Set}_G$ assigns to each object $x \in M$ (the space of objects of $\mathscr{P}$) a $G$-set $\Omega_x := W(x)$, and to each morphism $\sigma \colon x \to y$ in $\mathscr{P}$ a $G$-equivariant map $W(\sigma) \colon \Omega_x \to \Omega_y$. This functor field naturally incorporates the global symmetries described by the group $G$.  For instance, if $\Gamma$ itself were equivalent to the action groupoid $G \times X$ for some $G$-set $X$, then a functor $W\colon \mathscr{P} \to G \times X$ would encode how the system's configurations (related to $X$) transform under both the probing system's transitions (morphisms in $\mathscr{P}$) and the internal $G$-symmetry.

As before, considering additional topological properties leads us to consider a topological $G$-bundle $\pi\colon \Omega \to M$, that is, a topological bundle with structure group $G$.   If we impose further restrictions and we ask that the group $G$ acts freely on the sets of the category, we will have that a functor $W \colon \mathscr{P} \to \mathrm{Aut\,}_G(\pi)$ will be a representation of $\mathscr{P}$ on the topological principal bundle $\pi \colon \Omega \to M$ with structure group $G$.    In such case the functor $W$ is commonly called a (topological left-) action of the groupoid $\mathscr{P}$ on the $G$-principal bundle $\pi \colon \Omega \to M$. 

Thus, the ``representation'' perspective highlights how the probing system $\mathscr{P}$ acts on the structured space $\Omega$ describing the outcomes of the quantum system via the transitions encoded in $\Gamma$, respecting the internal symmetries and structures of $\Gamma$.

%%%%%%%%%%
%%%%%%%%%%

%%%%%%%%%%

\subsubsection{A reconstruction theorem}\label{sec:reconstruction}

Often, the situations we are leading to consider involve the choice of groupoids $\mathscr{P}\rightrightarrows M$, carrying a smooth structure, in which case, provided that the smooth structure is compatible with the groupoid structure\footnote{The composition map and the inverse map of the groupoid must be smooth and, in addition, it is required that the source and target maps must be submersions \cite{Ma05}.}, we will say that $\mathscr{P}$ is a Lie groupoid.     Then a smooth functor $W \colon \mathscr{P} \to \mathrm{Aut\,}_G(\pi)$, where $\pi \colon \Omega \to M$ is a smooth $G$-principal bundle, will be called a smooth (left) action of $\mathscr{P}$ on $\pi \colon \Omega \to M$.

\bigskip

Any representation/functor $W \colon \mathscr{P} \to \Gamma$, of the groupoid $\mathscr{P}$ on the groupoid $\Gamma$ restricts to a representation $W_{x_0} \colon \mathscr{P}(x_0) \to \Gamma (W(x_0))$, where $\mathscr{P}(x_0)$ denotes the isotropy group of $\mathscr{P}$ at $x_0 \in M$, and similarly for $\Gamma (W(x_0))$.    The isotropy groups $\mathscr{P}(x)$ corresponding to points $x\in M$ in the same orbit are all isomorphic.   If we restrict the groupoid $\mathscr{P}$ to a given orbit $\mathcal{O} \subset M$, then the restriction $W_\mathcal{O}$ of the functor $W$ to the orbit $\mathcal{O}$, that is, to the groupoid $\mathscr{P}_\mathcal{O} \subset \mathscr{P}$, is determined by the restriction $W_{x_0}$ to a given point $x_0 \in \mathcal{O}$.    The converse is also true and this is the content of the following reconstruction theorem stated below.

\begin{theorem}\label{thm:reconstruction}
Given a group $G$, a connected groupoid $\mathscr{P} \rightrightarrows M$, and a group homomorphism $W_0 \colon \mathscr{P} (x_0) \to G$, there exists a principal $G$-bundle $\pi \colon \Omega \to M$ and a functor $W \colon \mathscr{P} \to \mathrm{Aut\, }_G(\pi)$, such that the restriction of $W$ to $\mathscr{P}(x_0)$ is $W_0$.   Such correspondence is one-to-one. 
\end{theorem}

\begin{proof}  Let us consider the set $\mathscr{P}_{x_0} \times G$, where $\mathscr{P}_{x_0} = \{ \alpha\colon x_0 \to x \}$ denotes the set of morphisms whose source is $x_0$, i.e., $\mathscr{P}_{x_0} = s^{-1}(x_0)$, and the action of the group $\mathscr{P}(x_0)$, that is the isotropy group of $\mathscr{P}$ at $x_0$, on the right defined by:
\begin{equation}\label{eq:action}
(\alpha, g) \cdot \gamma_0 = (\alpha \circ \gamma_0, W_0(\gamma_0^{-1}) g) \, ,
\end{equation}
where $\alpha \colon x_0 \to y \in \mathscr{P}_{x_0}$, $g \in G$, $\gamma_0 \colon x_0 \to x_0 \in \mathscr{P}(x_0)$. 

The group $\mathscr{P}(x_0)$ acts freely on $\mathscr{P}_{x_0}$, hence the action (\ref{eq:action}) on $\mathscr{P}_{x_0} \times G$ is free.   We consider the quotient space 
$$
\Omega := \mathscr{P}_{x_0} \times_{\mathscr{P}(x_0)} G = \mathscr{P}_{x_0} \times G/ \mathscr{P}(x_0) \, . 
$$  
The projection map $\pi \colon \Omega \to M$ is given by $\pi([\alpha, g]) = t(\alpha)$.   It is well-defined because if $(\alpha', h') \sim (\alpha, h)$, then $\alpha' = \alpha \circ \gamma_0$ for some $\gamma_0 \in \mathscr{P}(x_0)$, and $t(\alpha') = t(\alpha \circ \gamma_0) = t(\alpha)$.

The group $G$ acts on $\Omega$ on the right via $[\alpha, g]\cdot h = [\alpha, gh]$, where $[\alpha, h]$ denotes the equivalence class containing $(\alpha,h)$. The action is well defined because if $(\alpha',h') \in [\alpha,h]$, then $\alpha' = \alpha \circ \gamma_0$ and $h' = W_0(\gamma_0^{-1})h$. Then $[\alpha', h']\cdot g = [\alpha \circ \gamma_0, W_0(\gamma_0^{-1})hg] = [\alpha, hg] = [\alpha, h] \cdot g$.  Notice that this action is free on the fibres $\Omega_x = \pi^{-1}(x)$.   In fact, if $[\alpha,h] \in \Omega_x$, then $t(\alpha) = x$.   Let $[\alpha,h] \in \Omega_x$ such that $[\alpha,h] \cdot g = [\alpha ,h]$. Then there exists $\gamma_0 \in \mathscr{P}(x_0)$ such that $\alpha = \alpha\circ \gamma_0$, and $hg = W_0(\gamma_0^{-1})h$.   But $\gamma_0 = 1_[x_0]$, and $hg = h$, that is $g = e$.

The functor $W \colon \mathscr{P} \to \mathrm{Aut\,}_G(\pi)$ is defined as follows: For a morphism $\beta \colon y \to z$ in $\mathscr{P}$, we  define $W(\beta) \colon \Omega_y \to \Omega_z$ as 
\begin{equation}\label{eq:W_definition}
W(\beta) ([\alpha, g]) = [\beta \circ \alpha, g] \, ,
\end{equation}
(notice that  an element in $\Omega_y = \pi^{-1}(y)$ is of the form $[\alpha, g]$ where $\alpha \colon x_0 \to y$).

We check that (\ref{eq:W_definition}) is well-defined. Suppose $[\alpha', g'] = [\alpha, g]$. Then $\alpha' = \alpha\circ \gamma_0$ and $g' = W_0(\gamma_0^{-1})g$ for some $\gamma_0 \in \mathscr{P}(x_0)$.
$W(\beta)([\alpha', g']) = [\beta \circ \alpha', g'] = [\beta \circ \alpha \circ \gamma_0, W_0(\gamma_0^{-1})g]$.
Is this equivalent to $W(\beta)([\alpha, g]) = [\beta \circ \alpha, g]$? Yes, because $(\beta\circ \alpha\circ \gamma_0, W_0(\gamma_0^{-1})g)$ is obtained from $(\beta\circ\alpha, g)$ by acting with $\gamma_0$ on the right according to (\ref{eq:action}). So $W(\beta)$ is well-defined.

The map $W(\beta)$ sends $\Omega_y$ to $\Omega_z$ because $t(\beta \circ\alpha) = t(\beta) = z$.
It is an automorphism in $\mathrm{Aut\,}_G(\pi)$ (i.e., a $G$-equivariant isomorphism covering the identity on $M$).  Indeed it respects the $G$-action: $W(\beta)( [\alpha, g]\cdot h) = W(\beta)([\alpha, gh]) = [\beta\circ \alpha, gh] = [\beta\circ\alpha, g] \cdot h =  (W(\beta)([\alpha, g]))\cdot h$.   It is invertible with inverse $W(\beta^{-1})$.

We check functoriality. For $\beta_1 \colon y \to z$, $\beta_2 \colon z \to w$:
$W(\beta_2 \circ \beta_1) ([\alpha, g]) = [(\beta_2 \circ \beta_1) \circ \alpha, g] = [\beta_2 \circ (\beta_1 \circ \alpha), g]$, and
$(W(\beta_2) \circ W(\beta_1)) ([\alpha, g]) = W(\beta_2) ( W(\beta_1) ([\alpha, g]) ) = W(\beta_2) ( [\beta_1 \circ \alpha, g] ) = [\beta_2 \circ (\beta_1 \circ \alpha), g]$, which  match.   Also, $W(1_y)([\alpha, g]) = [1_y \circ \alpha, g] = [\alpha, g]$, so $W(1_y) = \mathrm{id}_{\Omega_y}$.
Thus, $W$ is a functor. 

Finally, we check the restriction $W|_{\mathscr{P}(x_0)}$. Let $\beta_0 \in \mathscr{P}(x_0)$. $W(\beta_0)$ maps $\Omega_{x_0}$ to $\Omega_{x_0}$. Elements of $\Omega_{x_0}$ are $[\alpha_0, g]$ with $\alpha_0 \in \mathscr{P}(x_0)$. $W(\beta_0)([\alpha_0, g]) = [\beta_0 \circ\alpha_0, g]$. We can identify $\Omega_{x_0}$ with $G$ via the map $j: G \to \Omega_{x_0}$ given by $g \mapsto [1_{x_0}, g]$. The inverse maps $[\alpha_0, g]$ to $W_0(\alpha_0)g$. Under this identification, the action of $W(\beta_0)$ corresponds to the map $g \mapsto W_0(\beta_0) g$ on $G$. That is, $j^{-1}(W(\beta_0)(j(g))) = j^{-1}(W(\beta_0)[1_{x_0}, g]) = j^{-1}([\beta_0, g]) = W_0(\beta_0)g$. This shows that the action induced by $W$ on the fibre over $x_0$ corresponds precisely to the original homomorphism $W_0$. The uniqueness follows from the construction.
\end{proof}

\begin{corollary}\label{cor:top_smooth}
Under the conditions of Thm. \ref{thm:reconstruction}, if the groupoid $\mathscr{P}$, the group $G$ and the group homomorphism $W_0$ belong to the category of topological spaces (or alternatively to the category of smooth manifolds, diffeological spaces or smooth sets), the same holds for the bundle $\Omega$ and the functor $W$.
\end{corollary}

\begin{proof}
Note that if $\mathscr{P}$ is, for instance, a topological groupoid, $G$ is a topological group, and $W_0$ is continuous, the construction in the proof of Thm. \ref{thm:reconstruction} yields a continuous functor $W$.  Indeed, the quotient space $\Omega = \mathscr{P}_{x_0} \times_{\mathscr{P}(x_0)} G$ carries the quotient topology, and the projection map $\pi \colon \Omega \to M$ is a submersion, i.e. open continuous (because the target map $t$ is a submersion).  Moreover the action of $G$ on the right is continuous (because the composition on the group $G$ is), and the functor $W$ is continuous because composition on the groupoid $\mathscr{P}$ is continuous).
\end{proof}

\begin{remark} As indicated in Cor. \ref{cor:top_smooth} the reconstruction theorem, Thm. \ref{thm:reconstruction} can be stated in other categories of groupoids, e.g., Lie groupoids or, as it will be discussed later, in the category of diffeological (or smooth set) groupoids.   The previous theorem is an extension of a reconstruction theorem proved by Barret \cite{Ba91} that provides a characterization of gauge fields in terms of holonomy maps (see later, Sect. \ref{sec:yang_mills}).
\end{remark}

\begin{remark} Theorem \ref{thm:reconstruction} is also a variation of an imprimitivity theorem for groupoid representations (see \cite{Ib25b} and references therein).   Indeed Mackey's imprimitivity theorem establishes a correspondence between representations of a group $G$ and representations of certain subgroups $K$.   In our case, the isotropy group $\mathscr{P}(x_0)$ is a subgroupoid of the total groupoid $\mathscr{P}$ and the correspondence between the representation $W_0$ of $\mathscr{P}(x_0)$ and the representation $W = W_0\uparrow \mathscr{P}$ is but an instance of a universal ``induction'' procedure. 
\end{remark}

\begin{remark}  The one-to-one correspondence between representations of a connected groupoid $\mathscr{P}$ and those of an isotropy group $\mathscr{P}(x_0)$ indicates that $\mathscr{P}$ and $\mathscr{P}(x_0)$ are Morita equivalent (see, for instance, \cite{La06} and references therein).    Note that the groupoid $\mathscr{P}$ acts on the fundamental normal subgroupoid $\mathscr{P}_0$ of $\mathscr{P}$ on the left while $\mathscr{P}(x_0)$ acts on $\mathscr{P}_0$ on the right asserting the Morita equivalence between both.
\end{remark}

%%%%%%%%%%%%
%%%%%%%%%%%%

\subsection{Topological gauge fields}\label{sec:top_gauge}  In addition to the example of Feynman's quantum mechanics, Sect. \ref{sec:feynman}, where histories are identified with Feynman's paths which were shown to be functor fields according to Def. \ref{def:quantum field}, we discuss here how gauge fields will also fit into this conceptual framework. In particular we will analize first topological gauge fields leaving the general discussion of general Yang-Mills fields to Sect. \ref{sec:yang_mills}.

Let us consider now a system whose quantum structure is described by a Lie group $G$, for instance a finite or a compact group, which may be thought of as the simplest class of groupoids (together with the unit groupoids) describing quantum systems, Sect. \ref{sec:quantum}\footnote{We may think of a quantum system described by a group as a very simple system with only one output. Then the transitions of the system must necessarily form a group.}.     A spacetime background for the theory will be provided by an experimental setting capable of determining the local structure of a spacetime manifold $\mathscr{M}$ carrying a causal structure encoded in a Lorentzian metric $\eta$.     In this section we will not take into account this additional structure and we will just assume that the auxiliary structure upon which our observations are build is just a smooth manifold $M$.   

We can also make the simplifying assumption that our probing system will be a groupoid $\mathscr{P} \rightrightarrows M$ whose transitions $\alpha\colon  x \to y$, will be determined by the trajectories $\gamma \colon [0,1] \to M$, of test particles moving in $M$, but our detection skills are blurred in such a way that continuous deformations of such trajectories cannot be distinguished, that is we cannot take ``plates'' or actual images of such trajectories, but our detectors are triggered only if a non-continuous change of the parameters describing such trajectory would happen\footnote{Think, for instance on a two-slit experiment like situation, where there is a hole with a solenoid passing through it in a two-dimensional experimental configuration, particles can wander around the hole and our detectors jump provided that the particle turns around the hole.}.    Then, it would be natural to consider that the probing groupoid $\mathscr{P}$ is just the homotopy groupoid $\pi_1 (M) \rightrightarrows M$, that is the morphisms $\alpha \colon x \to y$, $x,y \in M$ are homotopy classes $[\gamma]$ of pahts $\gamma\colon [0,1] \to M$, such that $\gamma (0) = x$, $\gamma(1) = y$.   The composition law is the standard composition of homotopy classes of paths denoted as $[\gamma_1] *[ \gamma_2]$.

According to the analysis in Sect. \ref{sec:probing}, our system will be described by a functor $W \colon \pi_1(M) \to \Gamma$, with $\Gamma$ a groupoid describing the inner quantum structure of the system.  But now, the restriction of the functor $W$ to the isotropy group of $\pi_1(M)$ at the base point $x_0\in M$, will define a group homomorphism $W_{x_0}$ from the  first Poincar\'e group $\pi_1(M,x_0)$ to the corresponding isotropy group of $\Gamma$, that must be the group $G$ we consider describes the quantum system we are probing.
Note that the isotropy group at $x_0$ of the homotopy groupoid $\pi_1(M)$ is just the first homotopy group $\pi_1(M,x_0)$ of $M$ based at $x_0$. If $M$ is connected, the groupoid $\pi_1(M)$ is transitive and all isotropy groups are isomorphic.   Thus, by choosing a reference point $x_0$, the restriction of $W$ to $x_0$, will define a homomorphism of groups $W_{x_0} \colon \pi_1(M,x_0) \to G$.

Moreover, from the proof of Thm. \ref{thm:reconstruction}, such group homomorphism determines a principal $G$-bundle $\Omega$ over $M$ given by $\Omega = \pi_1(M)_{x_0} \times G/\pi_1(M,x_0)$, where $\pi_1(M)_{x_0}$ denotes the set of all homotopy classes of paths $\sigma\colon x_0 \to x$, and $\pi_1(M,x_0)$ acts on the right as 
$$
([\sigma],g)[\gamma] = ([\sigma] *[ \gamma], W_{x_0}([\gamma]^{-1}) g) \, .
$$   
The natural projection $\pi \colon \Omega \to M$ is given by $\pi([\sigma,g]) = \sigma(1) = x$.  The functor $W$ is determined as in Thm. \ref{thm:reconstruction}, and we may conclude this quick analysis saying that a functor field (or history) $W$ of the probing system $\pi_1(M)$ will be a groupoid homomorphism $W \colon \pi_1(M) \to \mathrm{Aut\,}_G(\Omega)$, that is an assignment $[\gamma] \mapsto W([\gamma]) \colon \Omega_x \to \Omega_y$, such that 
\begin{equation}\label{eq:holonomy}
W([\gamma_1] * [\gamma_2]) = W([\gamma_1])\circ W([\gamma_2]) \, .
\end{equation}  

The homomorphism of groups $W_{x_0}$ determines a $G$-principal bundle $\Omega$ over $M$ equipped with a flat connection, up to isomorphism.  Specifically, there is a well-known correspondence: isomorphism classes of flat principal $G$-bundles over a path-connected, paracompact space $M$ are in one-to-one correspondence with conjugacy classes of homomorphisms $H \colon \pi_1(M, x_0) \to G$ (see, e.g., \cite{Ko96}). The homomorphism $W_{x_0}$ obtained from our functor field $W$ is precisely such a map. The functor $W \colon \pi_1(M) \to \mathrm{Aut}_G(M)$ itself represents the holonomy of this flat connection. That is, for any closed path $\gamma$ at $x_0$, the element $W([\gamma]) \in G$ is the group element obtained by parallel transporting along $\gamma$ using the flat connection associated with $W_{x_0}$. The functoriality condition Eq. (\ref{eq:holonomy}), becomes the defining property of holonomy.

Therefore, in this setup, a functor field $W \colon \pi_1(M) \to \mathrm{Aut}_G(M)$ is mathematically equivalent to specifying a flat $G$-connection on $M$. The ``field'' captures the global topological information about how the internal $G$ quantum system (we may also think of the group $G$ as the inner symmetry of the quantum system), twists around noncontractible loops in the manifold $M$. This provides a natural interpretation for topological gauge theories (like $BF$ theory in certain limits, or Chern-Simons theory related concepts) within the fields are functors framework, where the probing system $\pi_1(M)$ explicitly captures the topological obstructions of the base manifold.

%%%%%%%%%%%%%
%%%%%%%%%%%%%

\section{The categorical geometry of local field theories}\label{sec:category}

\subsection{Local properties of fields:  Sheaves}\label{sec:sheaves}

After we have established that fields can be understood as functors  defined on a probing system $\mathscr{P}$, we would like to qualify better the properties of such functor as fundamental building blocks of physical theories.      

Locality is a fundamental property of field theories that must be incorporated in our description of fields from the categorical viewpoint.    Let us discuss briefly the appropriate categorical notion of locality we are going to use in our treatment of fields as functors.    The proper notion of locality (and glueing) in categorical terms is provided by the notion of sheaves. 

A presheaf $F$ on the category $\mathsf{C}$ with values in the category $\mathsf{D}$ is a functor  $F \colon \mathsf{C}^\mathrm{opp} \to \mathsf{D}$ (also called a contravariant functor), where $\mathsf{C}^\mathrm{opp}$ denotes the opposite category to $\mathsf{C}$\footnote{The opposite category to a given category is the category with the same objects but morphisms reversed.}, that is $F(\alpha \circ \beta) = F(\beta) F(\alpha)$, for any pair of composable morphisms $\alpha, \beta$ in $\mathsf{C}$, and $F(1_c) = 1_{F(c)}$ for any object $c$ in $\mathsf{C}$.   

A sheaf $F$ is a presheaf satisfying a ``glueing property'' for ``local data''.    The best way to explain what ``local data'' and ``glueing'' means in this context is by discussing a prototypical example of sheaf, the canonical sheaf of a topological space.   Consider $X$ to be a topological space and, associated to it, we have the category $\mathsf{X}$ whose objects are open sets $U \subset X$ and whose morphisms are canonical embeddings $i \colon U \to V$, provided that $U \subset V$.   Thus the set of morphisms $\mathsf{X}(U,V) = \{ i \colon U \to V\}$, if $U \subset V$.  Note that $\mathsf{X}(U,V)$ is empty if $U$ is not a subset of $V$.   We can define the contravariant functor $F \colon \mathsf{X}^{\mathrm{opp}} \to \mathsf{Set}$, given by $F(U) := C(U)$, i.e., the functor that associates to any open set $U$ its commutative algebra of continuous functions on it, and if $i\colon U \to V$ is a morphism, then $F(i) \colon F(V) \to F(U)$, $F(i) f = f\circ i$.  The contravariant character of the functor $F$ reflects the localization property we were referring before, that is, if $U\to V$ is a morphism, then $F(V) \to F(U)$ localises the functor at $U$.  Moreover, the previous functor has an additional glueing property, namely, if $U_i, i \in I$, is a family of open sets covering $U$, that is $\bigcup_{i\in I} U_i = U$, then given continuous functions $f_i \colon U_i \to \mathbb{R}$, such that $f_i \mid_{U_i \cap U_j} = f_j \mid_{U_i \cap U_j}$, for all $i,j\in I$, there exists a continuous function $f\colon U \to \mathbb{R}$ such that $f\mid_{U_i} = f_i$.  In other words $f$ is obtained glueing the consistent local data $f_i$.    A presheaf $F$ on $\mathsf{X}$ satisfying the previous glueing condition is called a sheaf.    

There is a purely categorical description of sheaves beyond the specific examples provided by the canonical sheaves of topological spaces, smooth manifolds, sections of bundles, complex manifolds, algebraic varieties, etc.  In our work we will be concerned with sheaves defined on categories associated to groupoids describing probing systems.    The fundamental notion we need to define a sheaf, that is to understand the notion of ``glueing local data'', is that of \textit{coverings of objects} (that in the category $\mathsf{X}$ where just open coverings by open sets).   

In the category $\mathsf{X}$ associated to a topological space $X$, there are pull-backs, that is if $i \colon U \to V$ is a morphism (i.e., $U \subset V$) and $j \colon W \to V$ is another morphism (i.e., $W \subset V$), then there is an object, denoted $U\times_V W$, which is given by $U\cap W$, and morphisms $i_U \colon U \cap W \to U$ and $i_W \colon U \cap W \to W$, such that $i\circ i_U = j \circ i_W$, and such that they satisfy the universal property depicted in the folowing diagram (\ref{fig:pull-back}).  That is if $i' \colon O \to U$, and $j' \colon O \to W$, are morphisms such that $i\circ i' = j \circ j'$, then there exists a morphism $\nu \colon O \to U\times_V W = U\cap W$, such that 
\begin{equation}
i_U \circ \nu = i' \, , \quad  \mathrm{and\,}  \quad i_W \circ \nu = j' \, .
\end{equation}

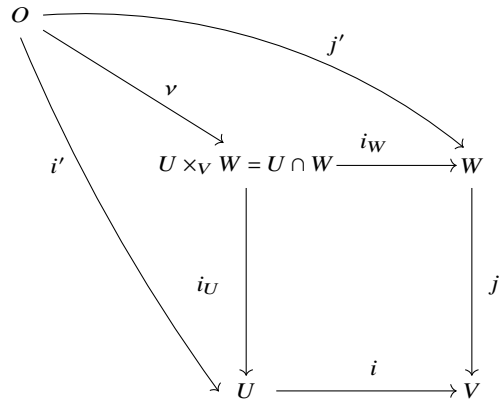
\begin{figure}\begin{center}
\begin{tikzpicture} 
\draw [->] (0.4,0) -- (2.8,0);
\draw [->] (0,2.7) -- (0,0.2);
\draw [->] (3,2.7) -- (3,0.2);
\draw [->] (1.2,3) -- (2.8,3);
\node  at (0,0)   {$U$};
\node at (3,0)   {$V$};
\node  at (0,3) {$U\times_V W = U\cap W$};
\node  at (3,3)   {$W$};
\node  at (1.7,0.3) {$i$};
\node  at (1.7,3.3) {$i_W$};
\node  at (-0.5,1.4) {$i_U$};
\node  at (3.3,1.4) {$j$};
\node  at (-3,5) {$O$};
\node  at (-1,4) {$\nu$};
\node  at (-2.5,3) {$i'$};
\node  at (1.2,4.6) {$j'$};
\draw [->] (-2.7,4.8) -- (-0.3,3.3);
\draw [->] (-2.7,5) arc (94:51:8);
\draw [->] (-3,4.7) arc (22.6:35.9:-23.3);
\end{tikzpicture}
\caption{Universal property defining pull-backs.}\label{fig:pull-back}
\end{center}
\end{figure}

In a general categorical setting the technical notion needed for the definition of a sheaf is a \textit{site}.  A site is a category $\mathsf{C}$ where we have ``coverings'', but if we want to extend the notion of open coverings given above, our category must have pull-backs, that is objects satisfying the universal property that defines the pull-back of open sets described in Fig. \ref{fig:pull-back}.   Then, coverings are neatly described as a set of families of morphisms $\{ \phi_i \colon c_i \to c\}$ in the category $\mathsf{C}$, each one of them called a covering (or ``coverage'' of the object $c$ of $\mathsf{C}$),  satisfying the following axioms \cite{Ar62}.

\begin{definition}\label{def:coverings}
Let $\mathsf{C}$ be a category.  A family of morphisms $\{ \phi_i \colon c_i \to c\}$, for each object $c$ in the category $\mathsf{C}$, are called coverings if they satisfy:
\begin{enumerate}
\item[i.] If $\phi \colon c' \to c$, is an isomorphism in $\mathsf{C}$, then it defines a covering of $c$.
\item[ii.] If $\{ \phi_i \colon c_i \to c\}$, and $\{ \phi_{ij} \colon c_{ij} \to c_i \}$ are coverings (of $c$ and $c_i$, respectively), then the family $\{ \phi_i \circ \phi_{ij} \colon c_{ij} \to c\}$ is also a covering of $c$.
\item[iii.] If $\{\phi_i \colon c_i \to c \}$ is a covering of $c$ and $\psi \colon d \to c$ is an arbitrary morphism, then the family of pull-back morphisms $\phi_i\times_c\psi \colon c_i\times_c d \to d$, is a covering of $d$ (see Fig. \ref{fig:coverages}).
\end{enumerate}
\end{definition}

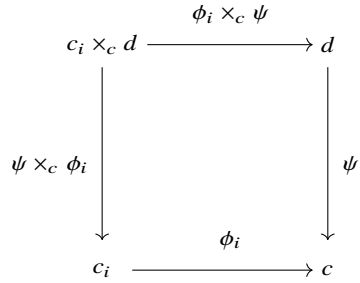
\begin{figure}\begin{center}
\begin{tikzpicture} 
\draw [->] (0.4,0) -- (2.8,0);
\draw [->] (0,2.7) -- (0,0.4);
\draw [->] (3,2.7) -- (3,0.4);
\draw [->] (0.6,3) -- (2.8,3);
\node  at (0,0)   {$c_i$};
\node at (3,0)   {$c$};
\node  at (0,3) {$c_i \times_c d$};
\node  at (3,3)   {$d$};
\node  at (1.7,0.4) {$\phi_i$};
\node  at (1.7,3.4) {$\phi_i\times_c \psi$};
\node  at (-0.7,1.4) {$\psi\times_c\phi_i$};
\node  at (3.3,1.4) {$\psi$};
\end{tikzpicture}
\caption{The pull-back condition, condition (iii) defining coverings.}\label{fig:coverages}
\end{center}
\end{figure}

Clearly the family of open coverings in the category $\mathsf{X}$ associated to a topological space $X$ define coverings in the previous sense as well as the families of open coverings in smooth manifolds, algebraic varieties, etc.  In other words the category $\mathsf{X}$ is a site whenever $X$ is a topological space, smooth manifold, etc.   

The family of functors $F \colon \mathsf{C} \to \mathsf{D}$ defined on a category $\mathsf{C}$ and taking values on the category $\mathsf{D}$ form a category, denoted $[\mathsf{C}, \mathsf{D}]$ whose objects are the functors themselves and the morphisms are natural transformations $\eta \colon F \Longrightarrow G$ (see later, Sect. \ref{sec:gauge_invariance}).   The category of presheaves $F \colon \mathsf{C}^\mathrm{opp} \to \mathsf{Set}$, will be denoted by $\mathrm{PrSh}(\mathsf{C})$, and, finally, the category of sheaves on the site $\mathsf{C}$ will be denoted by $\mathrm{Sh}(\mathsf{C})$.       A category of sheaves $\mathrm{Sh}(\mathsf{C})$ is called a Groethendieck topos (or just a topos) and they have a rich categorical structure \cite{Jo02}.   The category of sheaves on a category like $\mathsf{X}$ is called a \textit{petit} topos, in contrast to large categories of sheaves like the category of sheaves on the site $\mathsf{Eucl}$ whose objects are open sets in Euclidean spaces $\mathbb{R}^n$, $n \in \mathbb{N}$, and morphisms smooth maps $\varphi \colon U \to V$,  which are called \textit{gros} topos.  A sheaf $F \colon \mathsf{Eucl} \to \mathsf{Set}$ is called a smooth set and the topos of smooth sets (and the category of diffeological spaces which is a full subcategory of the topos of smooth sets) constitute the natural setting to discuss the geometry of spaces of fields (see \cite{Sc24,Gi23,Ib25a} and references therein).
 
The structure of groupoids provides a natural notion of coverings.   Thus given a groupoid $\mathscr{P}\rightrightarrows M$, we will consider the category $\mathsf{P}$, whose objects are subgroupoids $\mathscr{K}\subset \mathscr{P}$, and whose morphisms are the natural inclusions $\mathscr{H} \subset \mathscr{K}$.   

Given a groupoid $\mathscr{P}$, there is a natural family of subgroupoids $\mathscr{P}_U \subset \mathscr{P}$, where $U \subset M$ is a subset, and $\mathscr{P}_U \rightrightarrows U$ denotes the restriction of the groupoid $\mathscr{P}$ to $U$, in other words,
\begin{equation}\label{eq:restriction}
\mathscr{P}_U = \{ \alpha \colon x \to y \in \mathscr{P} \mid x,y \in U \} \, .
\end{equation}

\begin{definition}\label{def:locally}
We will say that a groupoid $\mathscr{P}\rightrightarrows M$ is locally generated if given any subset $U \subset M$, there is a covering $U_i$, $i\in I$ of $U$, $U = \cup_{i\in I} U_i$, such that the groupoid $\mathscr{P}_U$ is generated by the subgroupoids $\mathscr{P}_{U_i}$ associated to the covering $U_i$ of $U$.
\end{definition}

Note that there is no reference to a topological struture on $M$.   If $M$ were a topological space (or a smooth manifold) we may consider open subsets and open coverings.  In such case,  it is easy to check that if $M$ is compact, the groupoid $\mathscr{P}$ is locally generated.  
 
 \begin{definition}\label{def:groupoid_coverings}
Given a groupoid $\mathscr{H} \rightrightarrows U$ over $U \subset M$, which is an object in the category $\mathsf{P}$, that is, a subgroupoid of $\mathscr{P}$, a covering of $\mathscr{H}$ will be a family of subgroupoids $\mathscr{H}_{U_i} \subset \mathscr{H}$ associated to a covering $U_i$, $i \in I$ of $U$ by sets $U_i \subset U$, $i\in I$, such that $\mathscr{H}_{U_i}$ generates $\mathscr{H}$.  
\end{definition}

\begin{lemma}\label{lem:site}
    In the category $\mathsf{P}$ associated to a locally generated groupoid $\mathscr{P}$ there are coverings, i.e., the category $\mathsf{P}$ is a site.
\end{lemma}
 
 \begin{proof}
     
 Now it is easy to check the properties defining coverings, Def. (\ref{def:coverings}), for the families of subgroupoids $\mathscr{H}_{U_i}$ defining coverings in the category $\mathsf{P}$.  

 Properties 1 and 2 are trivial.  Property 3 in Def. \ref{def:coverings}, follow easily realizing that the pullback in the category $\mathsf{P}$ is given by intersections of groupoids and that if $\mathscr{K} \subset \mathscr{H}$ is a subgroupoid of $\mathscr{H}\rightrightarrows U$, and $U_i$ is a covering of $U$ such that the subgroupoids $\mathscr{H}_{U_i}$ generate $\mathscr{H}$, then the groupoids $\mathscr{H}_{U_i}\cap \mathscr{K}$ generates $\mathscr{K}$. Indeed, notice that $\mathscr{H}_{U_i}\cap \mathscr{K} = \mathscr{K}_{U_i}$, and if $\alpha \colon x \to y \in \mathscr{K} \subset \mathscr{H}$, then, $\alpha = \alpha_r \circ \cdots \circ \alpha_1$, where each factor $\alpha_k$ belongs to some $\mathscr{H}_{U_{i_k}}$, $k = 1, \ldots, r$.  Then, consider the factorization 
 $$
 \alpha =   \sigma_1 \circ \sigma_1^{-1} \circ  \alpha  \,,
 $$ 
 where $\sigma_1 \colon x_1 \to y  \in \mathscr{K}$, with $x_1 \in U_{i_1}$. Then,  
 $$
 \alpha_1' =  \sigma_1^{-1} \circ\alpha \colon x \to x_1 \in \mathscr{K}_{U_1} \, ,
 $$  
 and we get:
 $$
 \alpha = \sigma_1 \circ  \alpha_1'  \, .
 $$
 Repeating the process for $\sigma_1 \colon x_1 \to y$, we get $\sigma_1 = \sigma_2  \circ \sigma_2^{-1} \circ  \sigma_1 $, with $\sigma_2 \colon x_2 \to y \in \mathscr{K}$, $x_2 \in U_{i_2}$. Thus, $\alpha_2' = \sigma_2^{-1} \circ \sigma_1  \colon x_1 \to x_2 \in \mathscr{K}_{U_{i_2}}$ and 
 $$
 \alpha = \sigma_2  \circ \alpha_2' \circ \alpha_1' \, , 
 $$ 
 with $\alpha'_1 \in \mathscr{K}_{U_1},\alpha_2' \in \mathscr{K}_{U_2}$.   Iterating the process we obtain that $\alpha = \alpha_1' \circ \cdots \circ \alpha_r'$, with $\alpha_k' \in \mathscr{K}_{U_{i_k}}$.
 \end{proof}  

Thus, in this context, if $W \colon \mathscr{P} \to \Gamma$ is a functor field defined on a locally generated groupoid, it defines a groupoid homomorphism $W_\mathscr{H} := W\mid_{\mathscr{H}} \colon \mathscr{H} \to \Gamma$ for any $\mathscr{H} \subset \mathscr{P}$, and conversely, if $W_i \colon \mathscr{H}_i\to \Gamma$ is a family of groupoid homomorphisms defined on a covering $\mathscr{H}_i$ of $\mathscr{H}$, such that they coincide on intersections, that is, $W_i \mid_{\mathscr{H}_i \cap \mathscr{H}_j} = W_j \mid_{\mathscr{H}_i \cap \mathscr{H}_j}$, then there is groupoid homomorphism $W_\mathscr{H} \colon \mathscr{H} \to \Gamma$, given by $W_{\mathscr{H}}(\alpha) = W_{i_r}(\alpha_r ) \circ \cdots \circ W_{i_1}(\alpha_1)$, where $\alpha = \alpha_r \circ \cdots \circ \alpha_1$, with $\alpha_i \in \mathscr{H}_i$.  Notice that the restriction of $W_\mathscr{H}$ to $\mathscr{H}_i$ is $W_i$.

We can restate the previous discussion by saying that there is a contravariant functor $[-,\Gamma]$ defined in the category $\mathsf{P}$ with values in the category of groupoid homomorphisms with values in $\Gamma$, that is the functor assigns to each $\mathscr{H} \subset \mathscr{P}$, the category of groupoid homomorphisms $[\mathscr{H}, \Gamma]$.  Moreover, to each inclusion $i \colon \mathscr{K} \to \mathscr{H}$, it assigns the restriction $W \colon \mathscr{H} \to \Gamma \in [\mathscr{H}, \Gamma] \mapsto W\mid_{\mathscr{K}} = W\circ i \colon \mathscr{K} \to \Gamma \in [\mathscr{K}, \Gamma]$,  and this assignment is contravariant.

Thus, as observed above, the contravariant functor $[-,\Gamma]$ is a sheaf on the category $\mathsf{P}$. The sheaf property of the functor $[-,\Gamma]$, captures the locality and glueing properties of fields.

As discussed at the beginning of this section, coverings determine Groethendiek topologies \cite{Ar62} and a category with a Groethendiek topology is called a site.    Thus we can refine the previous definition of functor fields as quantum histories, Def. \ref{def:quantum field}, by making explicit a locality property.

\begin{definition}\label{def:local_field}
Given an inner quantum structure provided by a groupoid $\Gamma \rightrightarrows \Omega$ and a probing system described by the groupoid $\mathscr{P} \rightrightarrows M$, local quantum fields $W$ are elements of the sheaf $[-,\Gamma]$ defined on the site $\mathsf{P}$ associated to the probing groupoid $\mathscr{P}$. 
\end{definition}

There is a one-to-one correspondence between functor fields $W \colon \mathscr{P} \to \Gamma$ and global elements of the sheaf $[-,\Gamma]$.

%%%%%%%%%%%%%
%%%%%%%%%%%%%

\subsection{The category of local fields determined by a probing system: gauge invariance}\label{sec:gauge_invariance}

Given two categories $\mathsf{C}, \mathsf{D}$, the 
family of functors $F \colon \mathsf{C} \to \mathsf{D}$, also denoted as $[\mathsf{C}, \mathsf{D}]$, is a category whose objects are the functors $F$ themselves and whose morphisms are natural transformations $\varphi \colon F \Longrightarrow F'$, that is $\varphi$ assigns to any object $c \in \mathsf{C}$, a morphism $\varphi (x) \colon F(c) \to F'(c)$ in $\mathsf{D}$, such that 
\begin{equation}\label{eq:natural}
F'(\alpha) \circ \varphi (x) = \varphi (y) \circ F(\alpha) \, , \qquad \forall \alpha \colon x \to y \in \mathsf{C}(x,y) \, .
\end{equation}
 Our main observation here is that this notion captures the deep structure of gauge invariance in field theories.

Indeed, consider a probing system determined by the groupoid $\mathscr{P} \rightrightarrows M$, then functor fields are groupoid homomorphisms $W \colon \mathscr{P} \to \Gamma$, where $\Gamma \rightrightarrows \Omega$, describes a quantum system.   We can broaden our perspective and consider the class of all functors $W \colon \mathscr{P} \to \mathsf{Set}$, inside of which lies the family of functors $W \colon \mathscr{P} \to \Gamma$.   In fact the functors $[\mathscr{P}, \mathsf{Set}]$ form a category as commented above, and the subcategory of functors with codomain $\Gamma$ is often called the slice category of $[\mathscr{P}, \mathsf{Set}]$ at $\Gamma$.   

A natural transformation $\varphi \colon W \Longrightarrow W'$, will also be called an intertwiner between $W$ and $W'$ thought of as representations of the groupoid $\mathscr{P}$, Sect. \ref{sec:representations}.  In fact, two representations $R,R'$ of a given groupoid are equivalent if there is a invertible natural transformation $\varphi \colon R \Longrightarrow R'$.

If the natural transformation $\varphi \colon W \Longrightarrow W'$ is invertible we will say that the two functor fields are gauge equivalent and we will also call the natural transformation $\varphi$ a ``gauge transformation''.   Thus, two gauge equivalent functor fields describe equivalent representations of the probing system, and they may be considered equivalent from the perspective of the given probing system\footnote{But not necessarilly  physically equivalent in the sense of providing the same description of the theory.}.
Some examples will help clarify the previous notions.

\begin{example}
Consider Feynman's quantum mechanics as discussed before, Sect. \ref{sec:feynman}.  Fields are functors $W \colon P(\mathbb{R}) \to P(Q)$, then, a natural transformation $\varphi \colon W \Longrightarrow W'$, will assign to each $t\in \mathbb{R}$, a pair $\varphi(t) \colon \gamma (t) \to \gamma' (t)$, such that, Eq. (\ref{eq:natural}):
$$
\varphi (s) \circ W (s,t) = W' (s,t) \circ \varphi (t) \, ,
$$ 
which is identically satisfied.   Then, from the perspective of the clock considered as a probing system all paths are ``gauge equivalent'', i.e., they provide equivalent representations of the groupoid $P(\mathbb{R})$, whovever, different paths are not, in general, physically equivalent.  That will depend on the dynamical content of the theory provided, for instance by a Lagrangian function.   Thus, if the dynamics is described by a mechanical Lagrangian $L \colon TQ \to \mathbb{R}$, $L (q,v) = \frac{1}{2} \langle v,v\rangle_q - V(q)$, with $\langle \cdot,\cdots \rangle_q$ a metric on $Q$, then, the only natural transformations preserving the dynamics will correspond to isometries of the metric locally defined around a given path $q(t)$.   On the other hand if our Lagrangian is first order, for instance $L(q,v) = A_q(v)$, where $A = A_i dq^i$ is a 1-form on $Q$, then the equations of motion become first order equations on $Q$ of the form $\dot{x} \lrcorner F_A = 0$, with $F_A = dA$. Then, a natural transformations of a path $q(t)$ preserving the dynamics will correspond to a diffeomorphism $\varphi$ locally defined around $q(t)$ such that $\varphi^* A = A + d\chi$, i.e., to a (local)  standard gauge transformation of $A$ thought as a vector potential.
\end{example}

\begin{example}
There is, however, a situation where the notion of equivalence of representations adquires a deep physical sense.  If we consider topological gauge field theories as discussed in Sect. \ref{sec:top_gauge}, then two representations, or holonomy maps, $W,W'\colon \pi_1(M) \to \mathrm{Aut}_G(M)$ of the groupoid $\pi_1(M)$, define two principal $G$-bundles $P,P'$ respectively, Thm. \ref{thm:reconstruction}.  A natural transformation $\varphi \colon W \Longrightarrow W'$, is the same as an intertwiner for the representations $W$ and $W'$, that is $x \mapsto \varphi_x := \varphi (x)\colon P_x \to P_x'$, such that, Eq. (\ref{eq:natural}):
$$
W'([\gamma]) \circ \varphi(x) = \varphi(y) \circ W([\gamma])\, , \qquad \forall \gamma \colon x \to y\, ,   \quad x,y\in M \, .
$$  
If $\varphi$ is invertible, the two bundles are isomorphic (and the flat connection $A,A'$ defined by each one of them are gauge equivalent in the standard sense).  

Notice, that the equivalence of the representations $W$, $W'$, imply the equivalence of the representations $W_{x_0}$ and $W_{x_0}'$ of the isotropy group $\pi_1(M, x_0)$.   Thus we may think of the space of holonomy maps $W$ on the space $M$ as a category whose objects are the representations themselves and whose morphisms are natural transformations. The canonical subgroupoid of this category will have as morphisms invertible natural transformations, i.e., gauge transformations.  The space of orbits of this groupoid is the moduli space of gauge equivalence classes of flat connections on principal $G$-bundles over $M$ which is the same as the quotient space $\mathrm{Hom\,}(\pi_1(M,x_0), G)/G$.
\end{example}

\subsection{Yang-Mills fields}\label{sec:yang_mills}

 Probing a system described by a pincipal $G$ bundle $\pi \colon \Omega \to M$, over $M$ (or, better its associated groupoid $\mathrm{Aut}_G(\pi)$) can be done consistently by using ``test particles''.   In Sect. \ref{sec:top_gauge}, the chosen probing system was the groupoid $\mathscr{P} = \pi_1(M)$ where the trajectories of the particles were completely blurred and only their topological properties were observed.   If we improve the quality of our detectors and assume that the actual trajectories of the particles can be observed (if we wished so), then we will be lead to considering the space of curves $\gamma \colon I \to M$, where $I \subset \mathbb{R}$ is an interval of time.   
 
 However, it is not possible to perform continuous measurements on a quantum system and preserve its quantum properties.  The Zeno effect prevents that (see, e.g., \cite{Fa08}).   Thus, in a genuine quantum environment, we should not consider parametrised curves but just the traces of the trajectories of ``test particles''.  As it turns out the most reasonable class of curves we can consider are thin homotopy equivalence classes of curves \cite{Ba91,Sc09}.  A thin homotopy between to curves $\gamma \colon [t_0,t_1]\to M$, $\gamma' \colon [t_0',t_1']\to M$, is a smooth homotopy $h \colon D \subset \mathbb{R}^2 \to M$, $h(s, \cdot)$, $s \in [0,1]$, such it that transforms  the map $\gamma (t)$ in the interval $[t_0,t_1]$,  in the map $\gamma'(t)$  in the interval $[t_0',t_1']$, and such that the rank of the differential map $Dh$ is smaller or equal than 1 (in other words, such that the deformation $\partial h /\partial s$ and the time derivative $\partial h /\partial t$ are parallel).   It is clear that thin homotopy classes of smooth curves include reparametrizations of the time parameter.   As it happens the quotient of the space of smooth paths $\Omega (M)$ by the thin homotopy equivalence relation $\sim_{\mathrm{thin}}$, can be equipped with a canonical groupoid structure.  Let us denote such quotient by $\mathscr{P}(M) = \Omega(M) / \sim_{\mathrm{thin}}$.   The composition law is defined as $[\gamma]*[\gamma'] = [\gamma \circ \gamma']$, where $\gamma \circ \gamma'$ denotes concatenation of curves and $[\gamma]$ denotes now the thin homotopy class of $\gamma$.  The source and target maps $s,t \colon \mathscr{P}(M) \to M$ are given by $s([\gamma]) = \gamma (0)$, $t([\gamma]) = \gamma (1)$ with $\gamma \colon [0,1] \to M$ a representative in the equivalence class $[\gamma]$.  
 
 The groupoid $\mathscr{P}(M) \rightrightarrows M$, carries a natural diffeological structure (it is indeed a concrete smooth space, see \cite{Ib25a} and references therein) and it makes sense to say that a map defined on $\mathscr{P}(M)$ is smooth.  A functor field on $W$ on $\mathscr{P}(M)$ will be a smooth functor $W \colon \mathscr{P}(M) \to \mathrm{Aut\,}_G(\pi_P)$, for some principal $G$-bundle $\pi_P \colon P \to M$.  The isotropy group of $\mathscr{P}(M)$ at $x_0\in M$ will be denoted by $\mathscr{P}(M,x_0)$ and it consists of all  thin homotopy classes of loops based at $x_0$.   Note that if the manifold $M$ is connected, the groupoid $\mathscr{P}(M)$ is transitive and all its isotropy groups are isomorphic.

 The restriction of a field $W$ on $\mathscr{P}(M)$ to the isotropy group $\mathscr{P}(M,x_0)$ defines a group homomorphism $W_{x_0} \colon \mathscr{P}(M,x_0) \to G$ (once a point $p \in P_{x_0}$ along the fibre of $x_0$ has been chosen).  Such homomorphism is smooth in the diffeological sense, what is equivalent to the technical condition H3 in the theorem by Barret \cite{Ba91} that establishes a one-to-one correspondence between group homomorphisms $H \colon \mathscr{P}(M,x_0) \to G$, satisfying a technical condition (that is exactly the condition stating that the homomorphism $H$ is smooth in the diffeological sense), and principal connections on a principal $G$-bundle with a marked point $p$ above $x_0$.   Barret's theorem provides the sought relation between holonomy maps associated to parallel transport and gauge fields represented by principal connections $A$ on principal bundles.   The group homomorphism $H$ is a holonomy map and the relation between the homomorphism $H$ and the principal connection $A$ associated to it is given by:
 \begin{equation}
 H([\gamma]) = P \exp \int_\gamma A \, ,
 \end{equation}
 where $P (\cdot)$ denotes the time-ordered exponential.

 Moreover, as a consequence of Thm. \ref{thm:reconstruction}, each homomorphism $H \colon \mathscr{P}(M,x_0) \to G$ determines a representation $W \colon \mathscr{P}(M) \to \mathrm{Aut\,}_G(\pi)$.  The bundle obtained by the reconstruction theorem Thm. \ref{thm:reconstruction}, coincides with the principal bundle obtained in Barret's construction, and the functor $W$ represents the parallel transport defining the connection $A$.  

 We may observe that the groupoid $\mathscr{P}(M)$ is generated locally, Def. \ref{def:locally}.   Indeed, if $U_i$ denotes a partition of $M$ by open subsets, and $\gamma \colon x \to y$ a path, then because of compactiy we can extract a finite family of open sets $U_{i_k}$, $k = 1, \ldots, r$, covering the graph of $\gamma$. Then we can decompose 
 $$
 [\gamma] = [\gamma_{i_r}] \circ \cdots \circ [\gamma_{i_1}] \, ,
 $$
where each of the paths $\gamma_{i_k}$ lies in the open subset $U_{i_k}$.   Then, the category of subgroupoids of the groupoid of paths $\mathscr{P}(M)$ is a site, Lemma \ref{lem:site}, and functor fields defined on it, i.e., holonomy maps, have nice local properties\footnote{We already know this because connection forms $A$ can be localized and reconstructed out of local data, but it is instructive to see that such property in encapsulated at a much deeper level.}. 

Finally, we observe that gauge transformations $\eta \colon W \Longrightarrow W'$ are exactly the same as the standard notion of gauge transformation among connections.   In fact, the natural transformation $\eta$ assigns an invertible equivariant map $\eta (x) \colon P_x \to P_x$ such that, Eq. (\ref{eq:natural}):
$$
\eta(y) \circ W([\gamma])= W'([\gamma])\circ \eta (x) \, ,
$$ 
for each path $\gamma \colon x \to y$, but this means that the parallel transports associated to $W$ and $W'$ are equivalent, or that the connection $A'$ is obtained from $A$ by the gauge transformation $\eta$.

The dynamical aspects and the quantum analysis of the theory of Yang-Mills fields from the perspective presented here will be discussed elsewhere.   It suffices to point out here the relevant connections emerging from the previous considerations and recent work on the subject (see, e.g., \cite{Sa24}, and references therein).

%%%%%%%%%%%%%%

\section{Higher geometry of field theories}\label{sec:higher}

In Sect. \ref{sec:sheaves}, Def. \ref{def:local_field}, local fields were defined as sheaves, that is, as functors satisfying natural locality properties (localization and glueing), however the property of glueing plays another relevant role as it allows to define a natural ``horizontal'' composition of quantum histories.   This horizontal composition of histories lifts to the category of functor fields on a given probing system, making it into a 2-category.  We will discuss these ideas in what follows.

\subsection{The horizontal composition of fields}

We will begin by recalling first the composition of histories on a one-dimensional quantum field theory, i.e., in Feynman's quantum mechanics, Sect. \ref{sec:feynman}, Eq. (\ref{eq:composition_feynman}).   In that setting histories (or fields) are just functors $W \colon  P(\mathbb{R})  \to \Gamma$, with $\Gamma$ the inner groupoid of the theory.   Now, we will consider the restrictions of a history $W$ to a given subinterval of time $[t_0,t_1]$, i.e., the restriction of $W$ to the subgroupoid $P[t_1,t_0] = [t_0,t_1] \times [t_0,t_1]$.   We will denote the restriction as $W_{t_1,t_0}$ (or just $W$ if the context is unambiguous).    Note that the history $W_{t_1,t_0} \colon [t_0,t_1] \times [t_0,t_1] \to \Gamma$ is completely described in terms of a single map 
$$
w_{t_1,t_0} \colon [t_0,t_1] \to \Gamma \, ,
$$ 
defined as $w_{t_1,t_0} (s) = W(s,t_0)$.   In fact, 
$$
W(s,t) = W(s,t_0) \circ W(t,t_0)^{-1} = w_{t_1,t_0}(s) \circ w_{t_1,t_0}(t)^{-1} \, ,
$$ 
for all $t,s \in [t_0,t_1]$.   Now, given two histories $W_{t_1,t_0}$ and $W'_{t_2,t_1}$ defined respectively on the groupoids $P[t_1,t_0] = [t_0,t_1] \times [t_0,t_1]$ and $P[t_2,t_1]= [t_1,t_2] \times [t_1,t_2]$, we can ``glue'' them together as a new history $W'_{t_2,t_1}\circ W_{t_1,t_0}$ defined on the groupoid $P[t_2,t_0] = [t_0,t_2] \times [t_0,t_2]$, and defined by Eq. (\ref{eq:composition_feynman}), or equivalently as the history determined by the function $\tilde{w}_{t_2,t_1} \colon [t_0,t_2] \to \Gamma$, defined as:  
\begin{equation}\label{eq:horizontal1}
\tilde{w}_{t_2,t_0}(s) = \left\{ \begin{array}{ll} 
w_{t_1,t_0}(s), & \mathrm{if\,} t_0 \leq s \leq t_1 \\ w'_{t_2,t_1}(s) \circ w_{t_1,t_0}(t_1) \,, & \mathrm{if\,} t_1 \leq s \leq t_2 
\end{array} \right.
\, .
\end{equation}
Namely, we define the history $W'_{t_2,t_1}\circ W_{t_1,t_0}$ as 
$$
W'_{t_2,t_1}\circ W_{t_1,t_0}(s,t) =  \tilde{w}_{t_2,t_0}(s) \circ \tilde{w}_{t_2,t_0}(t)^{-1} \, , \quad \forall s,t \in [t_0,t_2] \, .
$$ 
It is easy to check that this composition law is associative and has units.   The unit $1_{(x_0,t_0)}$ is the history defined in the interval $[t_0,t_0] = \{ t_0 \}$ by the map $w(t_0) = 1_{x_0}$, with $1_{x_0}$ the unit at $x_0$ in the groupoid $\Gamma$.    

The family of all histories with values in the groupoid $\Gamma\rightrightarrows \Omega$ form a category $\mathscr{C}(\Gamma) \rightrightarrows \Omega \times \mathbb{R}$.    This category, and its groupoidification, were used in \cite{Ci24} to provide a rigorous description of the Feynman formalism in the groupoid picture of quantum mechanics.   In that situation the inner groupoid $\Gamma$ was assumed to be a measure groupoid and the histories, were assumed to be defined by measurable functions.  Oher choices can be made depending on the nature of the groupoid $\Gamma$.   For instance, in Feynman's standard presentation of the path integral description of quantum mechanics, the histories were defined by continuous paths $\gamma \colon [t_0,t_1] \to \mathbb{R}^3$.   

The previous ideas can be extended to include field theories described on a spacetime $\mathscr{M}$.  Thus, we will assume that $\mathscr{M}$ denotes a $m = (1+d)$--dimensional globally hyperbolic spacetime.  We want to describe a quantum system given by the inner groupoid $\Gamma\rightrightarrows \Omega$ using probing systems $\mathscr{P} \rightrightarrows \mathscr{M}$ over the spacetime $\mathscr{M}$.  We will consider the simplest non-trivial probing system defined by the groupoid of pairs $P(\mathscr{M}) = \mathscr{M} \times \mathscr{M} \rightrightarrows \mathscr{M}$, of $\mathscr{M}$\footnote{Note that it coincides with the groupoid $\pi_1(\mathscr{M})$ if $\mathscr{M}$ is simply connected.}.  Then a history or functor field $W$ will be given by a functor $W \colon P(\mathscr{M}) \to \Gamma$. The section of the bundle $\pi \colon \Omega \to \mathscr{M}$, defined by $W$ will be denotes as usual by $\phi\colon \mathscr{M} \to \Omega$, $x \in \mathscr{M} \mapsto \phi(x) \in \Omega$.    

We can localize $W$ to any (open) subset $O\subset \mathscr{M}$.   However we will be interested in localizing our histories to ``blocks'' $\mathscr{M}_{21} \subset \mathscr{M}$, determined by two Cauchy hypersurfaces $\Sigma_1$, $\Sigma_2$ in $\mathscr{M}$.   More precisely, consider two Cauchy hypersurfaces $\Sigma_1$, $\Sigma_2$ in $\mathscr{M}$ such that $\Sigma_1$ precedes $\Sigma_2$ (that is, $\Sigma_2$ lies in the future causal set $J^+(\Sigma_1)$ of $\Sigma_1$ and we  denote it as $\Sigma_1 \prec \Sigma_2$), then $\mathscr{M}_{21} = \{ J^+ (p)\cap J^-(q) \mid p \in \Sigma_1, q\in \Sigma_2 \}$.    We may consider the smooth submanifold $\mathscr{M}_{21}$ as a ``thick'' slice of the spacetime $\mathscr{M}$ with boundary $\partial \mathscr{M}_{21} = \overline{\Sigma}_1 \sqcup \Sigma_2$.    We will denote the restriction of the history $W$ to $\mathscr{M}_{21}$ as $W_{21}$.   

 Now we will consider the family of pairs $(\Sigma_*, \varphi_\Sigma)$, where $\Sigma_*$ denotes a pointed Cauchy hypersurface, that is, the hypersurface $\Sigma$ and a point $x_*$ on it, and $\varphi_\Sigma \colon \Sigma \to \Omega$ is a classical field, i.e., a section of the restriction of the bundle $\pi \colon \Omega \to \mathscr{M}$ to $\Sigma$.  Thus, given a history $W$, its restriction to the block $\mathscr{M}_{21}$, can be thought as a morphism in a new category with source the marked Cauchy hypersurface $\Sigma_1$ and the field $\varphi_1 = \phi\mid_{\Sigma_1}$, and target the pair $(\Sigma_2, \varphi_2 = \phi\mid_{\Sigma_2})$.  We will also denote the history $W_{21} \colon P(\mathscr{M}_{21}) \to \Gamma$, as a morphism: 
 $$
 W_{21}\colon (\Sigma_1, \varphi_1) \to (\Sigma_2,\varphi_2) \, .
 $$  
As indicated above, the restrictions $W_{21}$ (also called local histories or just histories) can be thought as actual morphisms of a category whose composition law is defined in close analogy to the situation in $(1+0)$-dimensions discussed above, Eq. (\ref{eq:horizontal1}).  

Given two histories $W_{21}\colon (\Sigma_1,\varphi_1) \to (\Sigma_2,\varphi_2)$, and $W'_{32}\colon (\Sigma_2,\varphi_2') \to (\Sigma_3,\varphi_3')$, defined on the blocks $\mathscr{M}_{21}$ and $\mathscr{M}_{32}$ respectively, determined by the Cauchy hypersurfaces $\Sigma_1 \prec \Sigma_2 \prec \Sigma_3$, and the restrictions $\varphi_1, \varphi_2, \varphi_3$ of the fields $\phi$, $\phi'$, to $\Sigma_1, \Sigma_2, \Sigma_3$ respectively, we can define a new history: 
$$
W'_{32} \circ W_{21}\colon  (\Sigma_1,\varphi_1) \to (\Sigma_3,\varphi_3) \, ,
$$ 
the composition of $W_{21}$ and $W_{32}'$, provided that the restriction of the classical fields $\varphi_2 = \phi\mid_{\Sigma_2}$ and $\varphi'_2 = \phi'\mid_{\Sigma_2}$ coincide, by means of the function $\widetilde{W}_{31} \colon \mathscr{M}_{21} \cup \mathscr{M}_{32} \to \Gamma$, defined as (compare with Eq. (\ref{eq:horizontal1})):

\begin{equation}\label{eq:horizontal}
\widetilde{W}_{31} (x) = \left\{ \begin{array}{ll} W_{21}(x,x_1)\,  & \mathrm{if\,} x \in \mathscr{M}_{21} \\  W'_{32} (x, x_2) \circ W_{21}(x_2, x_1)\,  & \mathrm{if \, } x \in \mathscr{M}_{32} \end{array}  \right. \, ,
\end{equation}
where $x_1,x_2$ denote the marked points in the hypersurfaces $\Sigma_1, \Sigma_2$ respectively.
Then, the composed history is defined as the functor 
$W'_{32} \circ W_{21}\colon P(\mathscr{M}) \to \Gamma$, given by: 
\begin{equation}\label{eq:comp_functor}
W'_{32} \circ W_{21} (y,x) = \widetilde{W}_{31}(y) \circ \widetilde{W}_{31}(x)^{-1} \, .
\end{equation}

Notice that because $W_{21} (x,x) = 1_{\phi(x)}$, then, the field $\tilde{\phi}$ corresponding to the composition $W_{32}'\circ W_{21}$ is given by $W_{32}' \circ W_{21} (x,x) = 1_{\tilde{\phi}(x)}$. Thus, if $x \in \mathscr{M}_{21}$, then due to (\ref{eq:horizontal}),
$$
W_{32}' \circ W_{21} (x,x) = W_{21}(x,x_0) \circ W_{21}(x,x_0)^{-1} = 1_{\phi (x)} \, ,
$$ 
and, if $x \in \mathscr{M}_{32}$, then 
\begin{eqnarray*}
&& W_{32}'\circ W_{21} (x,x) = \widetilde{W}_{31}(x) \circ \widetilde{W}_{31}(x)^{-1} \\ && = W_{32}' (x, x_2) \circ W_{21}(x_2,x_1) \circ W_{21}(x_2,x_1)^{-1} \circ W_{32}' (x,x_2)^{-1} = 1_{\phi'(x)} \, .
\end{eqnarray*}
Thus, the field $\tilde{\phi}$ corresponding to the composition is given by the field $\phi$ on $\mathscr{M}_{21}$ and the field $\phi'$ on $\mathscr{M}_{32}$.  Because of this we will denote such a field as $\phi \sqcup \phi'$ and we call it the union (or coproduct) of the fields $\phi$ and $\phi'$.

\subsection{The 2-category of functor fields over a spacetime}\label{sec:stacks}

The horizontal composition (\ref{eq:comp_functor}) defined on the space of local histories determined by Cauchy hypersurfaces determines a 2-category $\mathsf{C}(\mathscr{M})$, whose 0-cells are given by pairs $(\Sigma_*, \varphi_\Sigma)$ of pointed Cauchy hypersurfaces and local data, 1-cells are the local histories $W \colon (\Sigma_1,\varphi_1) \to (\Sigma_2, \varphi_2)$ as described above, and its 2-cells are natural transformations among histories $\eta \colon W \Longrightarrow W'$. The vertical composition of 2-cells is standard composition of natural transformations, that is if $\nu \colon W_1 \Longrightarrow W_2$, and $\eta \colon W_2 \Longrightarrow W_3$ are two natural transformations, then $\eta\circ_v \nu \colon W_1 \Longrightarrow W_3$, is the natural transformation: 
\begin{equation}\label{eq:vertical}
\eta\circ_v \nu (x) = \eta (\nu (x))\colon W_1(x) = \phi_1(x) \to W_3 (x) = \phi_3(x) \, .
\end{equation} 
The horizontal composition of the natural transformations $\nu \colon W_1 \Longrightarrow W_2$, and $\nu' \colon W_1' \Longrightarrow W_2'$, with $W_a \colon (\Sigma_1, \varphi_1) \to (\Sigma_2,\varphi_2)$, $W_a' \colon (\Sigma_2,\varphi_2) \to (\Sigma_3,\varphi_3)$, $a = 1,2$, is given by:
$$
\nu \circ_h \nu' \colon W_1' \circ W_1 \to W_2' \circ W_2 \, ,
$$ 
with 
$$
\nu \circ_h \nu' (x) \colon \phi_1 \sqcup \phi_1'(x) \to \phi_2 \sqcup \phi_2'(x) \, ,
$$ 
defined as: 
$$
\left\{ \begin{array}{ll} \nu \circ_h \nu' (x) = \nu (x) \colon \phi_1 (x) \to \phi_2 (x) \, , & \mathrm{if\,}  x \in \mathscr{M}_{21} \, ,\\
 \nu \circ_h \nu' (x) = \nu' (x) \colon \phi_1' (x) \to \phi_2' (x)\, , & \mathrm{if\,}  x \in \mathscr{M}_{32} \end{array}
 \right. \, . 
$$ 
Thus, $\nu \circ_h \nu'$ can also be written as: 
\begin{equation}\label{eq:2horizontal}
\nu \circ_h \nu' = \nu \sqcup \nu' \, .
\end{equation}

\begin{theorem}
$\mathsf{C}(\mathscr{M})$ is a 2-category.
\end{theorem}

\begin{proof}
It remains to be checked the exchange identity between the horizontal and vertical compositions:
\begin{equation}\label{eq:exchange}
(\eta\circ_h\eta') \circ_v (\nu\circ_h\nu') = (\eta\circ_v\nu) \circ_h (\eta'\circ_v\nu') \, .
\end{equation}

When, we apply the l.h.s. of (\ref{eq:exchange}) to an element $x\in \mathscr{M}_{21}$,  using Eqs. (\ref{eq:vertical}), (\ref{eq:2horizontal}), we get $\eta(\nu(x))$ which coincides with what we get when we apply the r.h.s. of (\ref{eq:exchange}) to $x\in \mathscr{M}_{21}$.
Similarly, If we apply the l.h.s. of (\ref{eq:exchange}) to $x \in \mathscr{M}_{32}$, we get $\eta'(\nu'(x))$, that, again, it agrees with the result of applying it to the r.h.s. of (\ref{eq:exchange}) due to Eqs. (\ref{eq:vertical}), (\ref{eq:2horizontal}).  See below, Fig. \ref{fig:exchange},  for an abstract diagrammatic proof of the exchange identity.

\begin{figure}[ht]
\centering
\begin{tikzpicture} 
%first blow
\fill (0.2,0) circle  (0.1);
\fill (3.5,0) circle  (0.1);
\draw [thick,->] (0.3,0.2) arc (160:20:1.7);
\draw [thick,->] (0.3,-0.2)  arc (20:160:-1.7);
\draw [thick,->] (0.4,0) -- (3.2,0);
\node  at (1.9,1.6)   {$W_1$};
\node  at (1.9,-0.2)   {$W_2$};
\node  at (1.9,-1.6)   {$W_3$};
\node at (-0.2,0.5)   {$(\Sigma_1,\varphi_1)$};
\node at (3.5,1.2)   {$(\Sigma_2,\varphi_2)$};
\node at (7.3,0.5)   {$(\Sigma_3,\varphi_3)$};
\node  at (1.9,1.6)   {$W_1$};
\node  at (1.9,0.7) {$\Big\Downarrow$};
\node  at (2.3,-0.7)   {$\eta$};
\node  at (1.9,-0.7) {$\Big\Downarrow$};
\node  at (2.3,0.7) {$\nu$};
\node  at (3,-2) {$\cong$};

%\node  at (3.9,0) {$\cong$};
% second blow
\fill (3.5,0) circle  (0.1);
\fill (6.8,0) circle  (0.1);
\draw [thick,->] (3.6,0.2) arc (160:20:1.7);
\draw [thick,->] (3.6,-0.2)  arc (20:160:-1.7);
\draw [thick,->] (3.7,0) -- (6.5,0);
\node  at (5.5,1.6)   {$W_1'$};
\node  at (5.5,-0.2)   {$W_2$};
\node at (5.2,-1.6)   {$W_3'$};
\node  at (5.2,0.7) {$\Big\Downarrow$};
\node  at (5.6,-0.7)   {$\eta'$};
\node  at (5.2,-0.7) {$\Big\Downarrow$};
\node  at (5.6,0.7) {$\nu'$};
\node  at (7.15,0) {$\cong$};
\end{tikzpicture}

\begin{tikzpicture} 
%first blow after composition
\fill (7.5,0) circle  (0.1);
\fill (10.8,0) circle  (0.1);
\draw [thick,->] (7.6,0.2) arc (140:40:2);
\draw [thick,->] (7.6,-0.2)  arc (40:140:-2);
%\node [left]  at  (-0.1,0)   {$v_1$};
%\node [right] at (4.4,0)   {$v_2$};
\node  at (9.2,1.3)   {$W$};
\node at (9.2,-1.3)   {$W_3$};
\node  at (9.2,0) {$\bigg\Downarrow$};
\node  at (9.9,0) {$\eta\circ_v \nu$};
\node  at (9.8,-2) {$\cong$};

%second blow after composition
\fill (10.8,0) circle  (0.1);
\fill (14.1,0) circle  (0.1);
\draw [thick,->] (10.9,0.2) arc (140:40:2);
\draw [thick,->] (10.9,-0.2)  arc (40:140:-2);
%\node [left]  at  (-0.1,0)   {$v_1$};
%\node [right] at (4.4,0)   {$v_2$};
\node  at (12.5,1.3)   {$W_1'$};
\node at (12.5,-1.3)   {$W_3'$};
\node  at (12.5,0) {$\bigg\Downarrow$};
\node  at (13.3,0) {$\eta'\circ_v \nu'$};
\end{tikzpicture}

%\newline

%first blow second line
\begin{tikzpicture} 
\fill (-0.3,0) circle  (0.1);
\fill (3,0) circle  (0.1);
\draw [thick,->] (-0.2,0.2) arc (160:20:1.7);
\draw [thick,->] (-0.2,-0.2)  arc (20:160:-1.7);
\draw [thick,->] (-0.1,0) -- (2.7,0);
\node  at (1.4,1.6)   {$W_1'\circ W_1$};
\node at (1.4,-1.6)   {$W_3'\circ W_3$};
\node  at (1.2,0.6) {$\Big\Downarrow$};
\node  at (2,-0.5)   {$\eta\circ_h \eta'$};
\node  at (1.2,-0.6) {$\Big\Downarrow$};
\node  at (2.05,0.5) {$\nu\circ_h \nu'$};

\node  at (3.9,0) {$\cong$};

%final blow second line
\fill (4.7,0) circle  (0.1);
\fill (8,0) circle  (0.1);
\draw [thick,->] (4.8,0.2) arc (140:40:2);
\draw [thick,->] (4.8,-0.2)  arc (40:140:-2);
%\node [left]  at  (-0.1,0)   {$v_1$};
%\node [right] at (4.4,0)   {$v_2$};
\node  at (6.4,1.3)   {$W_1'\circ W$};
\node at (6.4,-1.3)   {$W_3'\circ W_3$};
\node  at (6.4,0) {$\bigg\Downarrow$};
\end{tikzpicture}

\caption{A diagrammatic proof of the exchange identity: $(\eta\circ_h\eta') \circ_v (\nu\circ_h\nu') = (\eta\circ_v\nu) \circ_h (\eta'\circ_v\nu')$. Starting with the first diagram move either right down (first vertical composition), then jump to the bottom right (second horizontal composition), or jump to the bottom left (first horizontal composition) and then move right (second vertical composition).}
\label{fig:exchange}
\end{figure}
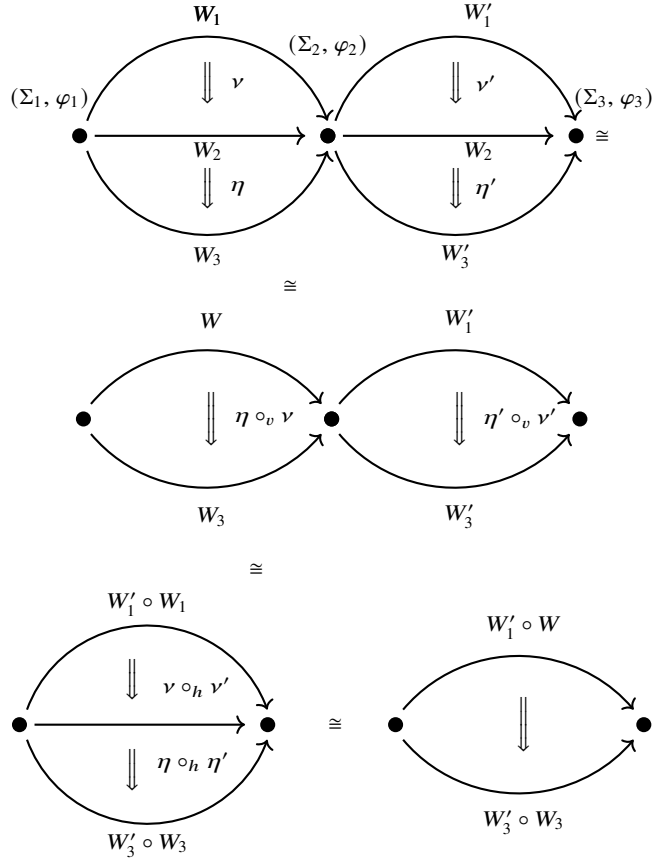

\end{proof}

%%%%%%%%%
%%%%%%%%%%
%%%%%%%%%%

\section{Conclusions and discussion}

A new categorical background for the notion of physical fields based on the groupoidal picture of quantum mechanics has been proposed.  The fundamental notion upon which the theory is constructed is that of probing systems, that is, the groupoidal description of ``test particles''.  Thus, in this picture, the physical fields are functors defined on groupoids describing quantum probing systems.   The representation theory of such functor fields has been derived and it has been shown that there is a one-to-one correspondence on each groupoid orbit between functor fields and group homomorphisms defined on the isotropy group of the probing groupoid. 

Local properties of functor fields have been explored finding that on locally generated probing groupoids, functor fields can be defined in terms of sheaves on the natural category associated to the probing groupoid.   

The notion of gauge transformations between physical fields is established as natural transformations among the corresponding functors.  The relevant instance of standard Yang-Mills fields is discussed from this point of view and the class of functor fields defined on a given probing system becomes a category itself.

The analysis of the theory of functor fields is concluded by succinctly discussing the horizontal composition of functor fields thought as ``histories'' when the probing system is defined over a globally hyperbolic spacetime.  Then the category of functor fields becomes a 2-category. 

This work constitutes the first approach to the new notion of functor fields.  Because the notion of functor fields is based upon the groupoidal description of quantum systems, this notion is inherently quantum.  However the examples discussed in the article are standard classical fields like sections of bundles and Yang-Mills fields.  We may argue that the standard fields corresponding to sections of a bundle are ``classical'' fields because the probing groupoid in such case can be thought as a classical system, however gauge fields are ``quantum'', meaning that the probing system is a system whose groupoid is non-classical (i.e., its algebra is non-commutatinve).  In any case, we will postpone the analysis of the quantum aspects of the theory presented here to subsequent work.  In particular we will discuss the relation of the notion of quantum fields in the Streater-Wightman axiomatic setting (see Sect. \ref{sec:introduction}) with the notion of fields developed here in subsequent papers.   The relation with Feynman's approach is simpler and is discussed in the main text: Feynman's paths are groupoid homomorphisms on clocks.  Further work will be devoted to analyzing the emergence of Haag-Kastler nets of algebras, i.e., the AQFT approach, from the perspective presented in this paper.

The role of symmetry and their functorial treatment is merely sketched in this paper and it will be the subject of further analysis in future papers \cite{Ib25b}. 

Finally, we must point out that no mention has been made to the dynamical description of the theory, that is, no dynamical principle for functor fields has been introduced.   A natural way to do it would be to extend previous work on the groupoidal setting of quantum mechanics, where the dynamics of the theory is introduced by means of a fusion of Feynman and Schwinger's principles (see \cite{Ci21b,Ib24,Ci24}).

\begin{acknowledgement}
A.I. and A.M. acknowledge financial support from the Spanish Ministry of Economy and Competitiveness, through the Severo Ochoa Program for Centers of Excellence in RD (SEV-2015/0554), the MINECO research project PID2020-117477GB-I00, and the Comunidad de Madrid project TEC-2024/COM-84 QUITEMAD-CM.
%%%
L.S. acknowledges financial support from Next Generation EU through the project 2022XZSAFN, PRIN2022 CUP: E53D23005970006.
%%%
L.S. is a member of the GNSGA (Indam).
\end{acknowledgement}
%

\input{references}

\end{document}

%% file: references.tex
%%%%%%%%%%%%%%%%%%%%%%%% referenc.tex %%%%%%%%%%%%%%%%%%%%%%%%%%%%%%
% sample references
% %
% Use this file as a template for your own input.
%
%%%%%%%%%%%%%%%%%%%%%%%% Springer-Verlag %%%%%%%%%%%%%%%%%%%%%%%%%%
%
% BibTeX users please use
% \bibliographystyle{}
% \bibliography{}
%
%\biblstarthook%}